\documentclass[11pt,letterpaper]{article}
\usepackage[margin=1in]{geometry}
\usepackage{amsfonts,amsmath,amsthm,graphicx,subcaption}
\usepackage{hyperref}
\newtheorem{theorem}{Theorem}

\newtheorem{observation}[theorem]{Observation}
\newtheorem{claim}[theorem]{Claim}

\newcommand\veclength[1]{{\left\|#1\right\|}}

\newcommand\Red{H}
\newcommand\Blue{P}
\newcommand\True{\textrm{true}}
\newcommand\False{\textrm{false}}
\newcommand\defn[1]{\textbf{\textit{\boldmath #1}}}

\title{Computational Complexity of \\ Flattening Fixed-Angle Orthogonal Chains%
  \thanks{A preliminary version of this paper was presented at the
    34th Canadian Conference on Computational Geometry.}}

\author{
  Erik D. Demaine%
    \thanks{Computer Science and Artificial Intelligence Laboratory, Massachusetts Institute of Technology, USA. \protect\url{{edemaine,jaysonl}@mit.edu}}
\and
  Hiro Ito%
    \thanks{School of Informatics and Engineering, The University of Electro-Communications, Japan. \protect\url{itohiro@uec.ac.jp}.
    Partially supported by JSPS Kakenhi Grant Number JP20K11671.}
\and
  Jayson Lynch\footnotemark[2]
\and
  Ryuhei Uehara%
    \thanks{School of Information Science, JAIST, Japan. \protect\url{uehara@jaist.ac.jp}.
 Partially supported by JSPS Kakenhi Grant Number JP18H04091, JP20H05961, JP20H05964, JP20K11673, and JP22H01423.}
}

\date{}

\begin{document}

\maketitle

\begin{abstract}
Planar/flat configurations of fixed-angle chains and trees are well
studied in the context of polymer science, molecular biology, and puzzles.
In this paper, we focus on a simple type of fixed-angle linkage:
every edge has unit length (equilateral),
and each joint has a fixed angle of $90^\circ$ (orthogonal) or $180^\circ$ (straight).
When the linkage forms a path (open chain),
it always has a planar configuration, namely the zig-zag
which alternating the $90^\circ$ angles between left and right turns.
But when the linkage forms a cycle (closed chain),
or is forced to lie in a box of fixed size,
we prove that the flattening problem ---
deciding whether there is a planar noncrossing configuration ---
is strongly NP-complete.

Back to open chains,
we turn to the Hydrophobic--Hydrophilic (HP) model of protein folding,
where each vertex is labeled H or P, and the goal is to find a folding that
maximizes the number of H--H adjacencies.
In the well-studied HP model, the joint angles are not fixed.
We introduce and analyze the fixed-angle HP model,
which is motivated by real-world proteins.
We prove strong NP-completeness of
finding a planar noncrossing configuration of a fixed-angle
orthogonal equilateral open chain with the most H--H adjacencies,
even if the chain has only two H vertices.
(Effectively, this lets us force the chain to be closed.)
\end{abstract}


\section{Introduction}

In this paper, we introduce and investigate a new model of protein folding.
We are given an \defn{equilateral fixed-angle chain} (``protein''),
where each vertex is marked H or P and has a specified fixed angle,
and edges all have unit length.
The goal is to embed the chain into a given grid
(e.g., 2D square, 3D cube, 2D triangular, or 2D hexagonal) while
\begin{enumerate}
\item respecting the fixed angles
  (but each angle is still free to be a left or right turn in 2D or spin in 3D);
\item avoiding self-crossing in the embedding; and
\item maximizing the number of H--H grid adjacencies.
\end{enumerate}
This is a fixed-angle version of the well-studied HP model of protein folding 
(where the angles are normally free to take on any value), 
which is known to be NP-hard in the 2D square grid \cite{HP-hard-2D}
and 3D cube grid \cite{HP-hard-3D}.
Fixed angles are motivated by real-world proteins;
see \cite[Chapters 8--9]{DO07}.
In the 2D square grid or 3D cube grid studied here, we can restrict to \defn{orthogonal} fixed-angle chains
where all fixed angles are $90^\circ$ or $180^\circ$.
For example, the popular ``Tangle'' toy restricts further to
all fixed angles being $90^\circ$; see \cite{DDHLTU2015}.

In the 3D cube grid, NP-hardness of fixed-angle HP protein folding
follows from \cite{AbelDemaineDemaineEisenstatLynchSchardl2013} 
which proves NP-hardness of embedding a fixed-angle orthogonal equilateral
chain of $n^3$ vertices into an $n \times n \times n$ 3D cube grid.
If we make all vertices Hs, then a cube embedding is the best way to maximize H--H adjacencies, 
as the cube uniquely minimizes surface area where potential adjacencies are lost.

In this paper, we prove that the fixed-angle HP protein folding problem
is NP-hard in the 2D square grid,
even if the chain has only two H vertices and those vertices are its endpoints.
In other words, given a fixed-angle orthogonal equilateral HP chain,
we prove it is strongly NP-hard to find any planar noncrossing embedding
where the endpoints (the two H vertices) are adjacent.
This result is tight in the sense that any fixed-angle orthogonal equilateral 
chain with fewer than two H vertices (and hence can have no H--H adjacencies)
has a noncrossing embedding, given by zig-zagging the $90^\circ$ angles
to alternate between left and right turns.

Fixed-angle HP protein folding where only the two endpoints are H vertices
is nearly equivalent to finding any planar noncrossing embedding of a
\defn{closed} fixed-angle chain
(where the first and last vertex are identified, 
and vertices are no longer marked H or P).
This is called the \defn{flattening problem} for fixed-angle closed chains.
The only difference is that, in the flattening problem, the first/last vertex
has a fixed-angle constraint, whereas in the HP model,
the two necessarily adjacent H vertices could form any angle.

Nonetheless, we show that the flattening problem for fixed-angle orthogonal
equilateral closed chains is strongly NP-complete.
Past work proved strong NP-hardness when this problem was generalized to
fixed-angle orthogonal equilateral caterpillar tree (instead of a chain)
or when we allow nonorthogonal fixed angles (and working off-grid) \cite{DemaineEisenstat2011}, but left this case open.

Finally our work also addresses two open problems from \cite{AbelDemaineDemaineEisenstatLynchSchardl2013}.
We solve one open problem by proving strong NP-completeness of deciding whether
a given fixed-angle orthogonal equilateral chain can be packed into a 2D square
(whereas \cite{AbelDemaineDemaineEisenstatLynchSchardl2013} proved an analogous result for a 3D cube).
We also prove that this problem remains NP-complete when
the chain is only a constant factor longer than the side length of the square
(and thus the square is sparsely filled), answering the 2D analog of a
3D question from \cite{AbelDemaineDemaineEisenstatLynchSchardl2013}.

\section{Preliminaries}

\subsection{Linkages}

A \defn{linkage} consists of a \defn{structure graph} $G=(V,E)$ and
edge-length function $\ell:E\rightarrow \mathbb{R}^+$.
A \defn{configuration} of a linkage in 2D is a mapping $C:V\rightarrow \mathbb{R}^2$ satisfying
the constraint $\ell(u,v)=\veclength{C(u)-C(v)}$ for each edge $\{u,v\}\in E$.
Let $x(C(u))$ and $y(C(u))$ be the $x$- and $y$-coordinate of $C(u)$, respectively.
A configuration is \defn{noncrossing} if any two edges $e_1,e_2\in E$ intersect only at a shared vertex $v \in e_1\cap e_2$.

A linkage is \defn{equilateral} if $\ell(e)=1$ for every $e\in E$.
A linkage with $n$ vertices is an \defn{open chain} if its structure graph $G$ is a path $(v_0,v_1,\ldots,v_{n-1})$,
and it is a \defn{closed chain} if $G$ is a cycle $(v_0,v_1,\ldots,v_{n-1},v_n=v_0)$.
A \defn{fixed-angle chain} is a chain together with
an angle function $\theta:V\rightarrow [0^\circ,180^\circ]$,
constraining configurations to have an angle of $\theta(v)$ at
every vertex $v$, except for the two endpoints of an open chain.
A fixed-angle chain is \defn{orthogonal} if we have $\theta(v_i) \in \{90^\circ,180^\circ\}$ for every vertex~$v_i$ with $0<i<n-1$.
A \defn{segment} of a fixed-angle chain
is a consecutive subchain $v_i, v_{i+1}, \dots, v_k$
with intermediate flat angles $\theta(v_j) = 180^\circ$ for $i < j < k$,
which acts the same as a single edge of length equal to the sum
($k-i$ for equilateral chains).

\label{sec:NP}
The \defn{embedding problem} asks to determine whether
a given linkage has a noncrossing configuration in 2D.
For general linkages, this problem is $\exists \mathbb R$-complete \cite{Kempe_SoCG2016}.
For fixed-angle orthogonal chains, the problem is in NP:
given a binary choice of turning left or right at each vertex,
we can construct an explicit embedding ---
placing the first vertex at the origin and the second vertex
on the positive $x$ axis, and adding and subtracting lengths
to the $x$ and $y$ coordinates --- and check for collisions and
(for closed chains) closure.
In fact, for fixed-angle orthogonal \emph{open} chains,
every instance is a ``yes'' instance:
\begin{observation}\label{obs:chain}
  Every fixed-angle orthogonal open chain has a noncrossing configuration.
\end{observation}
\begin{proof}
Intuitively, we embed the chain in a zig-zag.
Precisely, let $P=(v_0,v_1,\ldots,v_{n-1})$ be the path structure graph.
First we put $v_0$ at $(0,0)$, and $v_1$ at $(1,0)$.
For each $i=2,3,\ldots,n-1$, we define $x(C(v_i))$ and $y(C(v_i))$ as follows.
When $\theta(v_i)=180^\circ$, we have no choice: 
$x(C(v_i))=x(C(v_{i-1}))+(x(C(v_{i-1}))-x(C(v_{i-2})))$ and $y(C(v_i))=y(C(v_{i-1}))+(y(C(v_{i-1}))-y(C(v_{i-2})))$.
When $\theta(v_i)=90^\circ$ and $\overrightarrow{C(v_{i-2})C(v_{i-1})}$ is horizontal, 
we define $x(C(v_i))=x(C(v_{i-1}))$ and $y(C(v_i))=y(C(v_{i-1}))+1$. 
If it is vertical, we define $x(C(v_i))=x(C(v_{i-1}))+1$ and $y(C(v_i))=y(C(v_{i-1}))$.
The obtained configuration is noncrossing because it proceeds monotonically
in $x$ and $y$, with strict increase in one of the coordinates.
\end{proof}
We note that Observation \ref{obs:chain} holds for any fixed-angle orthogonal open chain 
which is not necessarily equilateral.

In the \defn{HP model}, the structure graph $G=(V,E)$ has its vertices
\defn{bicolored} by a color function $\omega:V\rightarrow\{\Red,\Blue\}$.
For a configuration $C$ of an equilateral orthogonal linkage,
a pair $(u,v)$ of vertices forms an \defn{H--H contact}
if $\omega(u) = \omega(v) = \Red$, $\veclength{C(u)-C(v)}=1$, and $\{u,v\} \notin E$.
The \defn{HP optimal folding problem} of a bicolored fixed-angle orthogonal
equilateral chain asks to find a noncrossing configuration of the linkage in 2D
that maximizes the number of H--H contacts.

\subsection{Linked Planar 3SAT}

Our reductions are from a variant of 3SAT, which we define and prove NP-hard
in this section.

In the standard \defn{3SAT} problem,
we are given a formula $\phi$ over a set $V$ of $n$ variables,
where $\phi$ is a conjunction of a set $C$ of $m$ \defn{clauses},
where each clause in $C$ is a disjunction of three \defn{literals},
where each literal is of the form $x$ or $\neg x$
for some variable $x \in V$.
In \defn{planar 3SAT}, we form the graph $G_\phi=(C\cup V,E)$
with a vertex for each variable in $V$ and each clause in $C$,
and edges between variables and the clauses that contain them,
and require that $G_\phi$ has a planar embedding.

We mimic a variant of planar 3SAT with additional planarity restrictions:
if we add edges to form a Hamiltonian cycle $\kappa$ of $C \cup V$ that
first visits all clauses in $C$ in some order $c_1,c_2,\dots,c_m$, and
then visits all variables in $V$ in some order $v_1,v_2,\dots,v_n$,
then the resulting graph $G'_\phi = G_\phi \cup \kappa$ must also be planar,
as in \figurename~\ref{fig:sat}.
The \defn{linked planar 3SAT problem} asks,
given $\phi$, $G_{\phi}$, $\kappa$, and a planar embedding of
$G_\phi \cup \kappa$, whether $\phi$ is satisfiable.
Pilz \cite{Pilz2018} proved this problem NP-complete.
Our proof follows the same high-level structure of this proof,
so we briefly review it now:

\begin{figure}[th]\centering
\begin{minipage}[t]{0.55\linewidth}\centering
\includegraphics[width=0.8\linewidth]{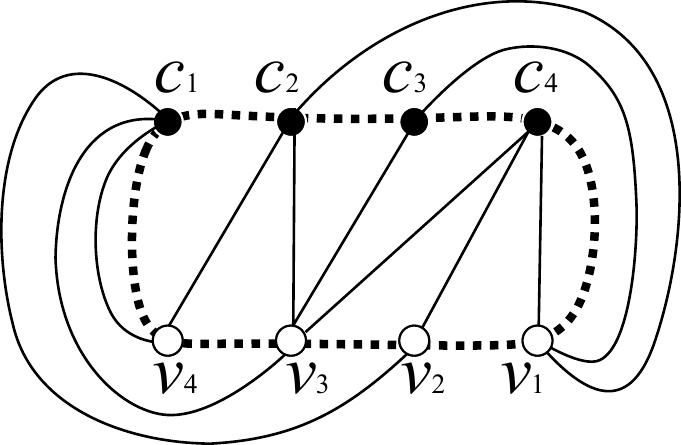}
\caption{An example instance of linked planar 3SAT,
where $c_1=(\neg v_2\vee \neg v_3\vee\neg v_4)$, $c_2=(v_4\vee v_3\vee \neg v_1)$,
$c_3=(\neg v_3\vee v_1)$, and $c_4=(v_1\vee v_2\vee v_3)$.
Hamiltonian cycle $\kappa$ (drawn dotted) visits $c_1,c_2,c_3,c_4,v_1,v_2,v_3,v_4$ in cyclic order.}
\label{fig:sat}
\end{minipage}
\hfil\hfill
\begin{minipage}[t]{0.35\linewidth}\centering
\includegraphics[width=0.8\linewidth]{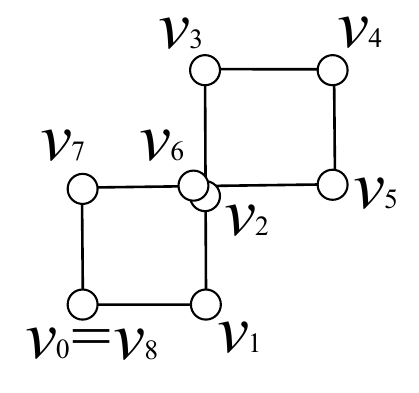}
\caption{A crossing chain $(v_0,v_1,v_2,v_3,v_4,v_5,v_6,v_7,v_8=v_0)$ with angles 
$\theta(v_2)=\theta(v_6)=180^\circ$ and $\theta(v_i)=90^\circ$ for $i=0,1,3,4,5,7$.}
\label{fig:cross}
\end{minipage}
\end{figure}

\begin{theorem} \label{thm:linked}
  Linked planar 3SAT is NP-complete.
\end{theorem}

\begin{proof}
  Pilz gives a reduction from planar 3SAT to linked planar 3SAT
  \cite{Pilz2018}.
  \figurename~\ref{fig:linked-3sat} summarizes the reduction.
  Let $G_\phi$ be the given planar 3SAT instance.
  First we draw the Hamiltonian cycle $\kappa$ (as on the left)
  in a large doubled spiral,
  with the inward spiral being the subpath for clauses
  and the outward spiral being the subpath for variables.
  Focusing on one half (the dotted square), we obtain a square
  with alternating horizontal grid lines for clauses and variables.
  Then we construct a planar drawing of $G_\phi$ with no horizontal edges,
  all variable vertices on odd grid lines (orange/light),
  and all clause vertices on even grid lines (purple/dark).
  This construction correctly orders the vertices on~$\kappa$.
  Finally we replace each edge of this drawing
  with a sequence of gadgets (as on the right)
  that duplicates the variable across the horizontal grid lines
  to reach the clause.
  Although not notated in our figure, each 4-cycle gadget alternates
  between positive and negative literals, so that for each variable $x$
  and copy $x'$ we have the clauses $\neg x \vee x'$ (i.e., $x \Rightarrow x'$)
  and $\neg x' \vee x$ (i.e., $x' \Rightarrow x$), which together imply $x=x'$
  (i.e., $x \Leftrightarrow x'$).
\end{proof}

\begin{figure}[th]\centering
\includegraphics[scale=0.5]{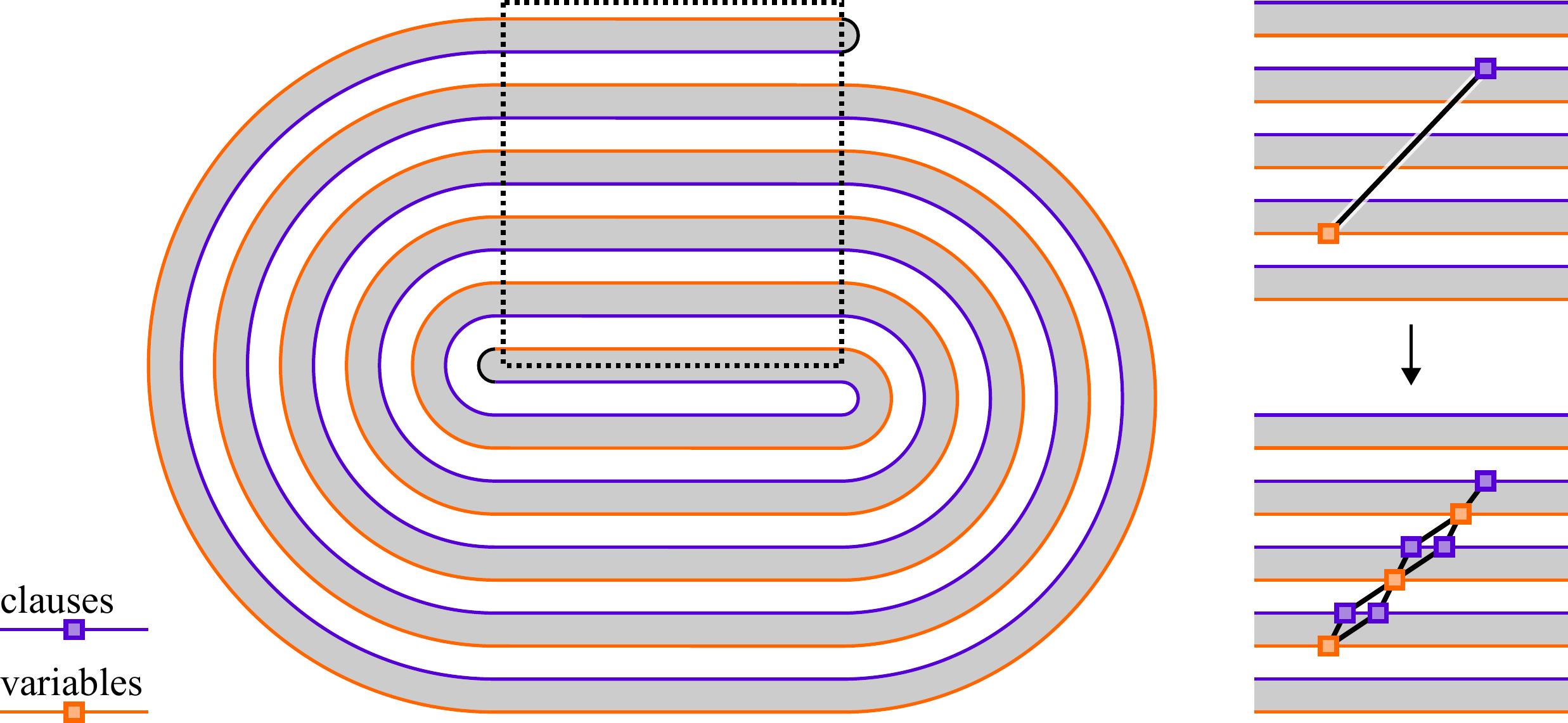}
\caption{Pilz's reduction from planar 3SAT to linked planar 3SAT
  \cite{Pilz2018}, based on Figs.~3 and~4 of~\cite{Pilz2018}.
  (For consistency with later figures, we use colors instead of dashes
  to distinguish variables and clauses, and rotate the figure $90^\circ$.)}
\label{fig:linked-3sat}
\end{figure}

\section{Embedding Fixed-Angle Orthogonal Equilateral Closed Chains is Strongly NP-complete}

In contrast to Observation \ref{obs:chain}, not all fixed-angle orthogonal equilateral \emph{closed} chains
are ``yes'' instances of the embedding problem.
In particular, an orthogonal equilateral closed chain must have an even
number of edges to have a configuration in 2D.
Even with this property, the length-$8$ chain $(v_0,v_1,v_2,v_3,v_4,v_5,v_6,v_7,v_8=v_0)$ with angles 
$\theta(v_2)=\theta(v_6)=180^\circ$ and $\theta(v_i)=90^\circ$ for $i=0,1,3,4,5,7$ has configurations in 2D but 
they have crossings at vertices $v_2$ and $v_6$ (\figurename~\ref{fig:cross}).
It is not difficult to show that the embedding problem for fixed-angle
orthogonal closed chains is weakly NP-hard by a reduction from the
ruler folding problem (see \cite[Chap.~2]{DO07});
this construction requires exponential edge lengths
(or equilateral chains with exponentially long straight segments).
In this section, we prove that the embedding problem is strongly NP-complete:

\begin{theorem}\label{th:loop}
  Embedding a fixed-angle orthogonal equilateral closed chain in 2D
  is strongly NP-complete.
\end{theorem}

Section~\ref{sec:NP} argued membership in NP.
To show NP-hardness, we mimic the reduction from planar 3SAT to
linked planar 3SAT given by Theorem~\ref{thm:linked}.
In particular, assume we have constructed a formula $\phi$,
the associated graph $G_\phi = (C\cup V,E)$, 
a Hamiltonian path $\kappa$ visiting $c_1,c_2,\ldots,c_m,v_1,v_2,\ldots,v_n$
in cyclic order,
and a planar embedding of $G_\phi \cup \kappa$ with $\kappa$ spiraling
as in \figurename~\ref{fig:linked-3sat}.
Figure~\ref{fig:sample-graph} shows an example of a planar 3SAT
instance and the result from this transformation,
and Figure~\ref{fig:sample} shows the final result of our reduction.

\begin{figure}\centering
\includegraphics[scale=0.5]{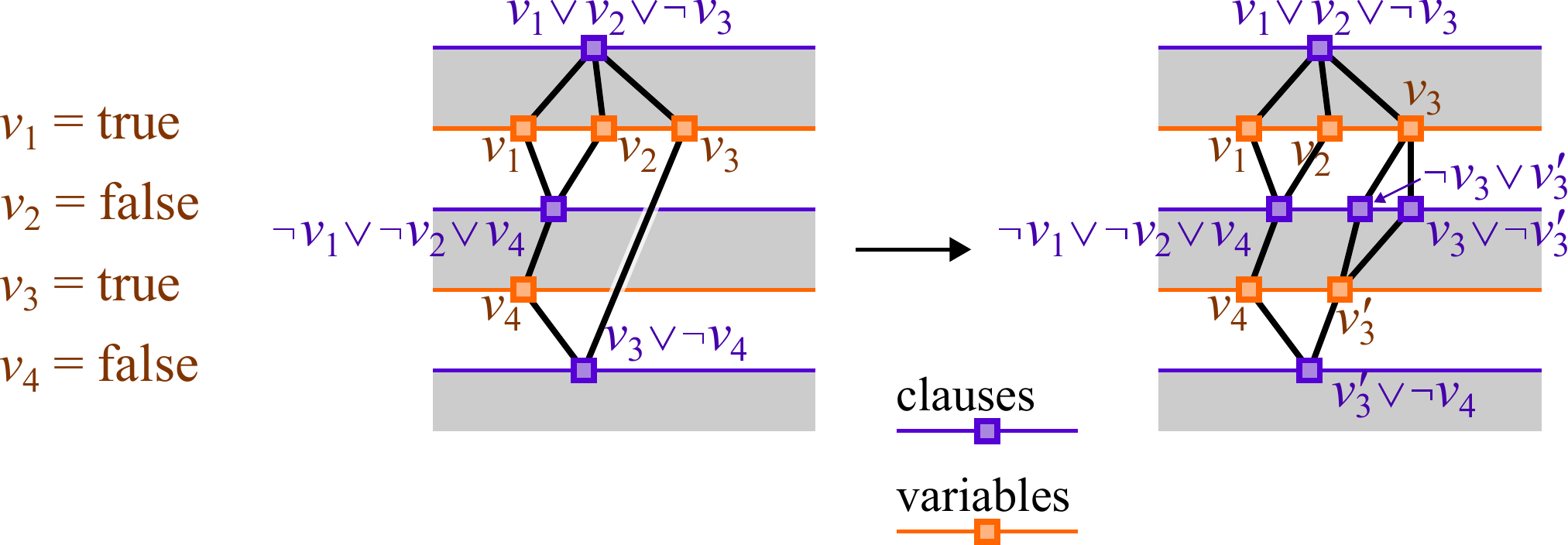}
\caption{Example planar 3SAT instance and the result of the transformation
  from Theorem~\ref{thm:linked}.}
\label{fig:sample-graph}
\end{figure}

\begin{figure}\centering
\includegraphics[width=0.77\linewidth]{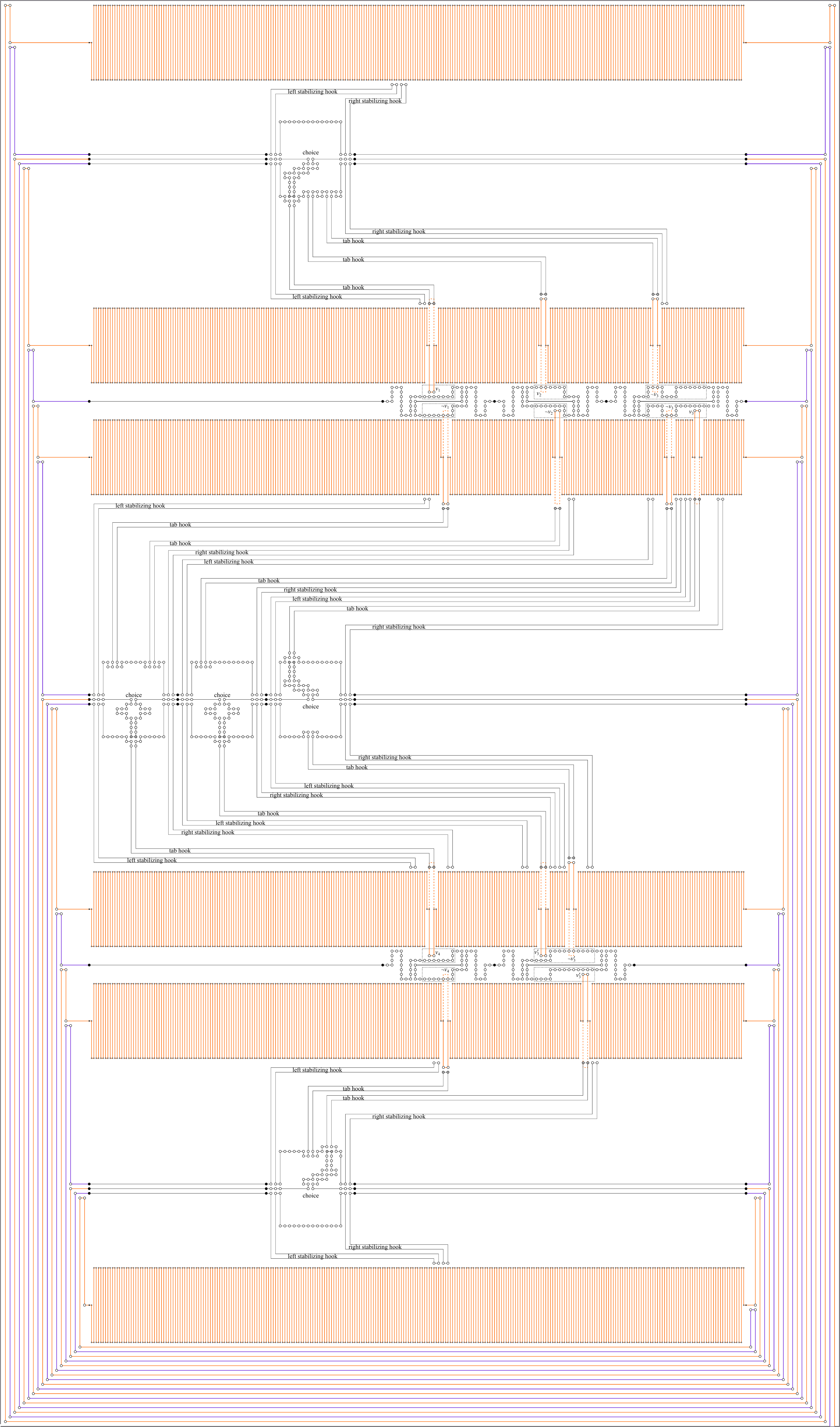}
\caption{An example of the reduction from the instance in
  \figurename~\ref{fig:sample-graph},
  and a solution embedding corresponding to
  assignment $v_1=\True$, $v_2=\False$, $v_3=\True$, and $v_4=\True$.
  (Insulation height and hook segments are not drawn to scale.)}
\label{fig:sample}
\end{figure}

\figurename~\ref{fig:overview} gives an overview of our construction.
Compared to \figurename~\ref{fig:linked-3sat}, we use the minimum vertical
space for the lower half of the rows, while we use significant (and varying)
vertical space for the upper half of the rows (where all the clauses and
variables are), to leave room for gadgets.
In addition, we have changed what the different rows are used for.
Rows cycle through sections containing different types of gadgets ---
``insulation'', variables,
more insulation, sheaths for clauses, choices for clauses,
sheaths for clauses --- starting and ending with insulation.
The chain thus visits alternating rows of variables and sheaths
followed by alternating rows of insulation and clause choices,
but it is different from $\kappa$ in a linked 3SAT instance
because it additionally visits rows of insulation and sheaths
in between variables and clauses that are connected together.
Finally, the entire construction is wrapped in a frame gadget
(shaded gray) which we show forces the drawn bounding box for the
entire construction.  We use this frame and the spiraling of the drawn
segments to force the endpoints of each row of gadgets to be as pictured.
Then we argue that the insulation gadgets force the locations, endpoints,
and bounding boxes of individual variable, sheath, and clause gadgets,
which allows us to argue their correctness.

\begin{figure}[t]\centering
  \includegraphics[scale=0.666]{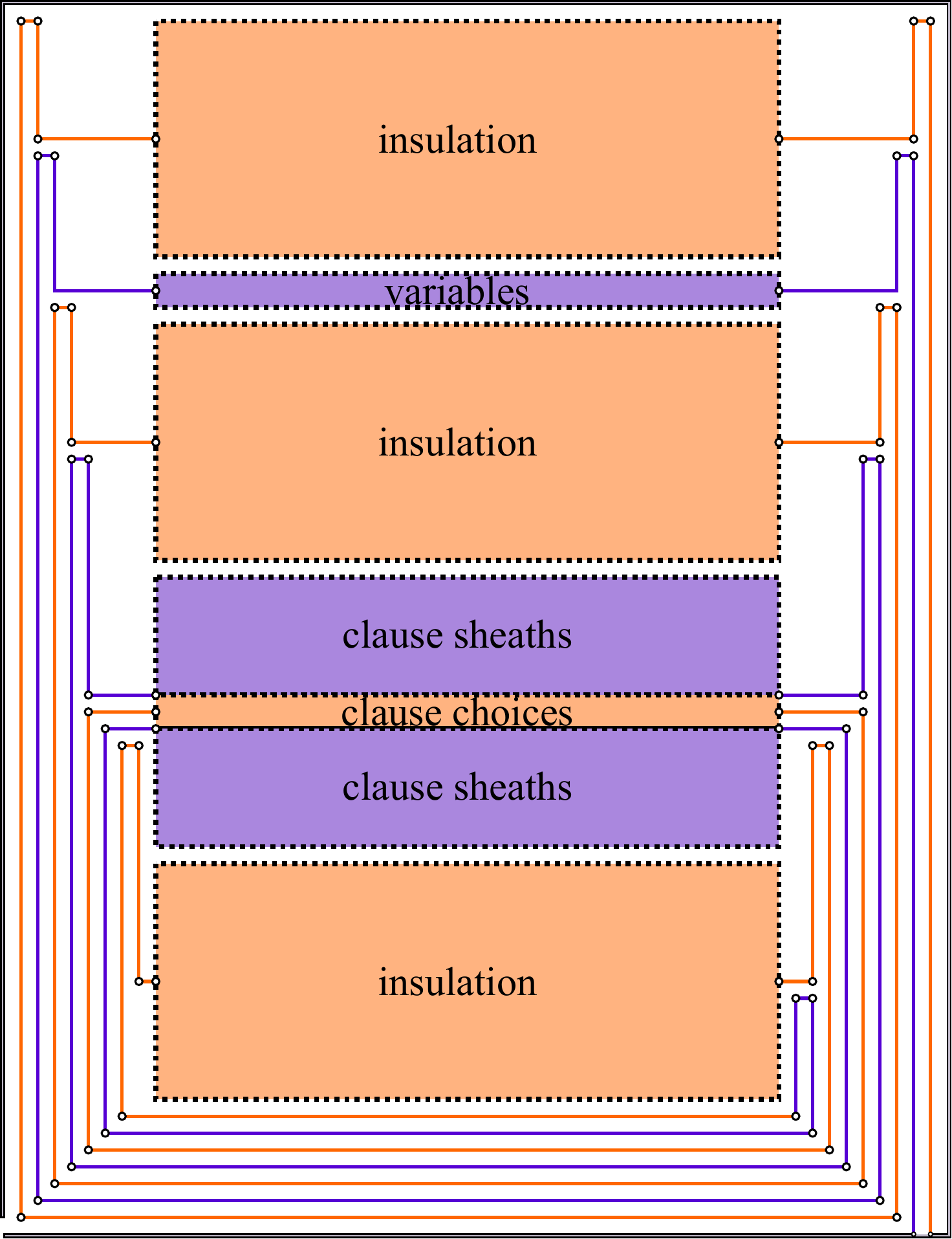}
  \caption{Construction overview, consisting of the outer frame gadget
    (shaded gray, on a grid $5 \times$ smaller than the rest),
    spiral gadget (colored segments), and
    rows of gadgets (dotted boxes)
    cycling through insulation gadgets, variable gadgets, sheath gadgets,
    choice gadgets, sheath gadgets, \dots, and ending with insulation gadgets.
    Insulation rows are extremely thick ($\Theta(n^2)$);
    sheath rows are medium thickness ($\Theta(n)$);
    variable rows are thin ($\Theta(1)$); and
    choice rows are extremely thin ($1$).
    (Straight vertices are not drawn to simplify the figure.)}
  \label{fig:overview}
\end{figure}

\subsection{Frame Gadget}

We start at the top level with the \defn{frame gadget},
shown in \figurename~\ref{fig:frame-loop}, which surrounds all other gadgets.
In fact, the frame gadget works with a $5 \times$ scaled version of the
rest of the construction, so that we can use parity arguments modulo~$5$.

\begin{figure}[th]\centering
  \includegraphics[scale=0.6]{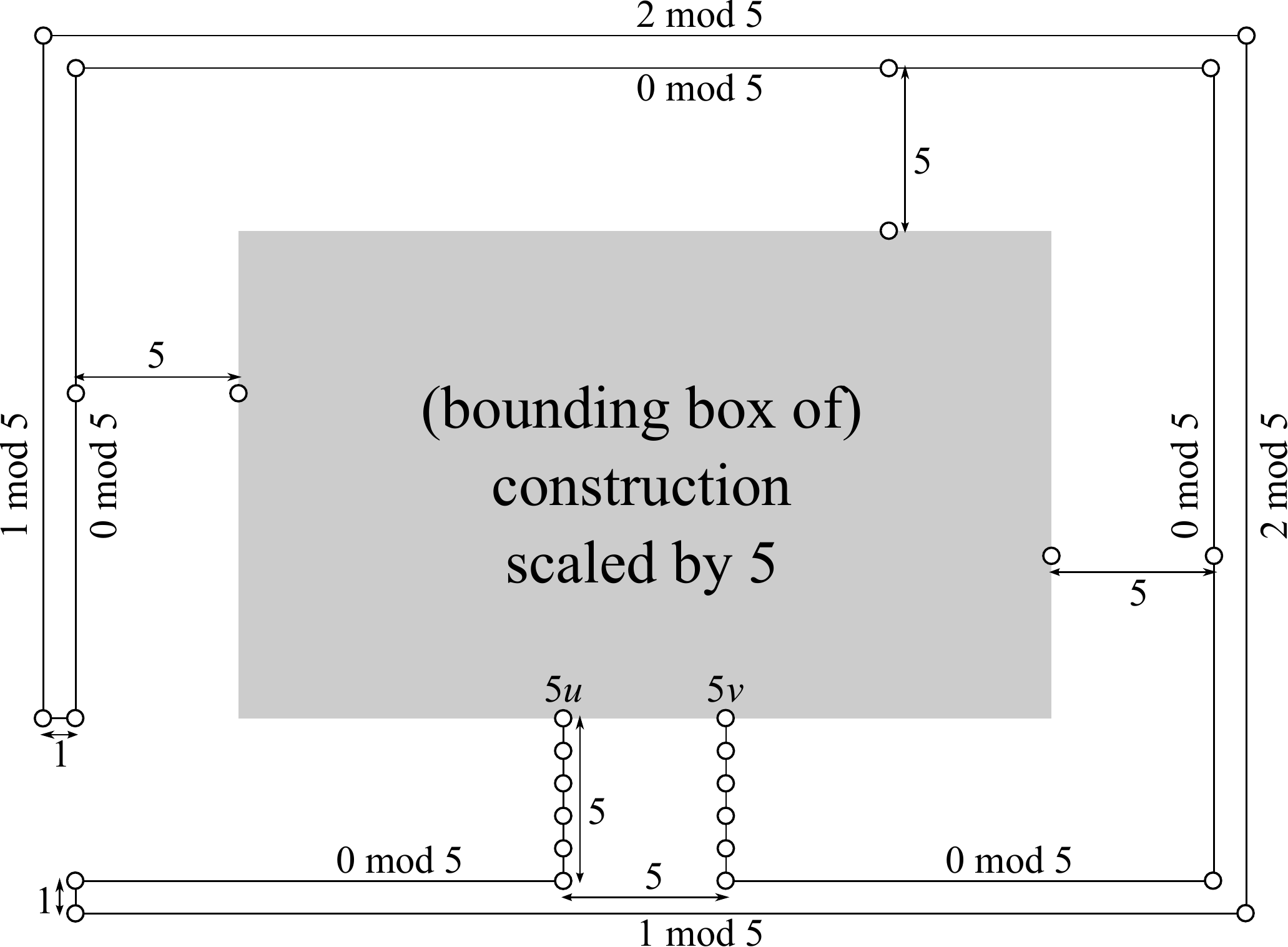}
  \caption{Frame gadget for closed chains,
    consisting of a frame construction~$F$ (all pictured edges)
    attached to a $5 \times$ scaled version of a given construction.
    (Some straight vertices are not drawn to simplify the figure.)}
  \label{fig:frame-loop}
\end{figure}

Formally, suppose we are given a closed chain~$C$,
a distinguished edge $\{u,v\}$,
an intended bounding box $B$ that the chain $C$ should fold into,
and two adjacent points $p,q$ on the bounding box that $u,v$
should fold to respectively.
We scale all inputs uniformly by a factor of~$5$
to produce a new chain $5 C$ with a distinguish edge $\{5 u, 5 v\}$
and a new bounding box $5 B$ with distinguished points $5 p$ and $5 q$
at distance $5$ from each other.
In particular, $5 C$ is the result of replacing each edge of the
original chain $C$ with a straight chain of five edges.
Thus, in $5 C$, all straight chains between $90^\circ$ vertices,
which we call \defn{segments}, have lengths equal to $0$ modulo~$5$.
Equivalently, the vertices of $5 C$ stay on a scaled square grid
with $5 \times 5$ cells.
We then modify $5 C$ by replacing the five-edge chain between $5 u$ and $5 v$
with the frame construction $F$ in \figurename~\ref{fig:frame-loop}
(i.e., all pictured edges),
with the intended folding building most of two rectangles $5$ and $6$
units away from the bounding box $5 B$.
Let $C'$ be the resulting chain (the combination of $5 C$ and $F$).

\begin{claim} \label{claim:frame-loop}
  The frame construction $F$ has a unique folding up to isometries,
  namely the one shown in \figurename~\ref{fig:frame-loop}.
\end{claim}

\begin{proof}
First observe that, once we embed one edge of the chain,
we know which segments are horizontal and vertical.
Thus we can assume by isometry that the horizontal/vertical assignment
is as in \figurename~\ref{fig:frame-loop}.

Because the chain forms a cycle, the signed vertical and horizontal distances
must sum to zero, and so in particular must sum to $0$ modulo~$5$.
The outer wrapping of the frame construction consists of three vertical
straight chains and three horizontal segments.
In each group of three, there are two segments of length $1$ modulo $5$
(one segment in fact has length exactly $1$),
and one segment of length $2$ modulo~$5$.
All other segments (including those in $5 C$) have length $0$ modulo~$5$.
The only solutions to $\pm 2 \pm 1 \pm 1 \equiv 0$ modulo $5$
are $2 - 1 - 1 \equiv 0$ and $-2 + 1 + 1 \equiv 0$.
Thus the $2$-modulo-$5$ segment in the group must have the opposite orientation
from the two $1$-modulo-$5$ segments in the group.
This forces the folding of the outer frame.

The rest of the frame construction has segment lengths that are $0$ modulo~$5$,
so it is impossible for it to transition between inside and outside of the
outer frame without collisions.
(The gap in the lower-left corner does not intersect any points on the
scaled square grid.)
Thus the rest of the chain must be inside:
if the next two segments both went outside, then they would immediately
collide;
and if one went outside and one went inside, then they could never meet.
Now the rest of the inner frame folding is forced because every segment length
is at least $5$, while the frame border has thickness $1$, so each consecutive
edge is forced to fold the only direction that remains inside the outer frame.
\end{proof}

\begin{claim} \label{claim:frame-contain}
  The framed construction $C'$ has a planar noncrossing embedding
  if and only if $C$ has a planar noncrossing embedding within~$B$.
\end{claim}

\begin{proof}
  This follows from Claim~\ref{claim:frame-loop} and that the scaled
  input chain $5 C$ cannot escape the frame (in the lower left) because
  it remains in a scaled square grid of coordinates $0$ modulo~$5$.
\end{proof}

Technically, the frame gadget comes last: the chain $C$ is the rest of the
construction (to be specified), and then our overall reduction is the
framed chain~$C'$.  By Claim~\ref{claim:frame-contain}, we can assume that
the rest of the construction $C$ is forced to remain within a desired
bounding box~$B$.

\subsection{Spiral Gadget}

The \defn{spiral gadget} consists of the colored segments in
Figure~\ref{fig:overview}, i.e., all drawn segments other than the frame gadget.
Other than the factor-$5$ scaling from the frame gadget,
the segments of the spiral gadget are drawn as tightly as possible,
on adjacent (scaled) grid lines.
For simplicity, we present the spiral gadget as if it is on a $1 \times 1$
grid instead of a $5 \times 5$ grid, and use one ``unit'' to mean
the length~$5$.
(In fact, this scale factor will increase to $10$
in Section~\ref{sec:insulation}.)

We use this tightness to argue that the spiral gadget has a unique folding,
even before specifying the details of the insulation, variable, and clause
gadgets.
The main property we need is that each dotted box in Figure~\ref{fig:overview}
(representing a row of insulation, variable, sheath, or choice gadgets)
is topologically a path,
connecting the two \defn{endpoints} shared with the spiral gadget
via some path of segments that starts and ends horizontal.
In addition, each dotted box has a specified width and height
(though the folded path of the gadget might not stay within that box),
all box widths are the same, and boxes are horizontally aligned.
(Boxes are also tightly packed vertically so that there is only one unit of
space between adjacent boxes, but we will not need this property.)
The heights of the upward vertical segments of the spiral
are then designed to be just below the endpoint (and spiral horizontal segment)
of the previous dotted box.
The spiral then typically U-turns to descend to the desired height
of the endpoint of the current dotted box;
in the special cases of choice rows and lower sheath rows,
the endpoint is at the top of the dotted box, so we omit the U-turn
and just turn immediately.
We require that the bottommost box is at least as high as all other boxes
(which follows from the bottommost box being insulation,
which is the tallest of all row types).

\begin{claim} \label{claim:spiral}
  The spiral gadget has a unique folding,
  pictured in Figure~\ref{fig:overview}.
  In particular, the endpoints where dotted boxes meet the spiral gadget
  have fixed positions relative to the frame.
\end{claim}

\begin{proof}
  Consider the segments $s_1, s_2, \dots, s_k$ of the spiral gadget,
  starting with the two extreme segments $s_1, s_k$ that attach to the frame;
  refer to Figure~\ref{fig:spiral-analysis}.
  These segments must both go upward, to avoid intersecting the frame.
  The outer (right) segment $s_1$ is maximally right in the bounding box $B$
  and goes all the way to the top,
  while the inner (left) segment $s_k$ is one grid line to the left of $s_1$
  and is a little shorter.
  For consistent labeling of segments, we assume that $s_1$ [and $s_k$]
  are immediately followed preceded] by a U-turn, but we also describe what
  modifications are necessary to handle the non-U-turn case.

  \begin{figure}[ht]
    \centering
    \includegraphics[scale=0.666]{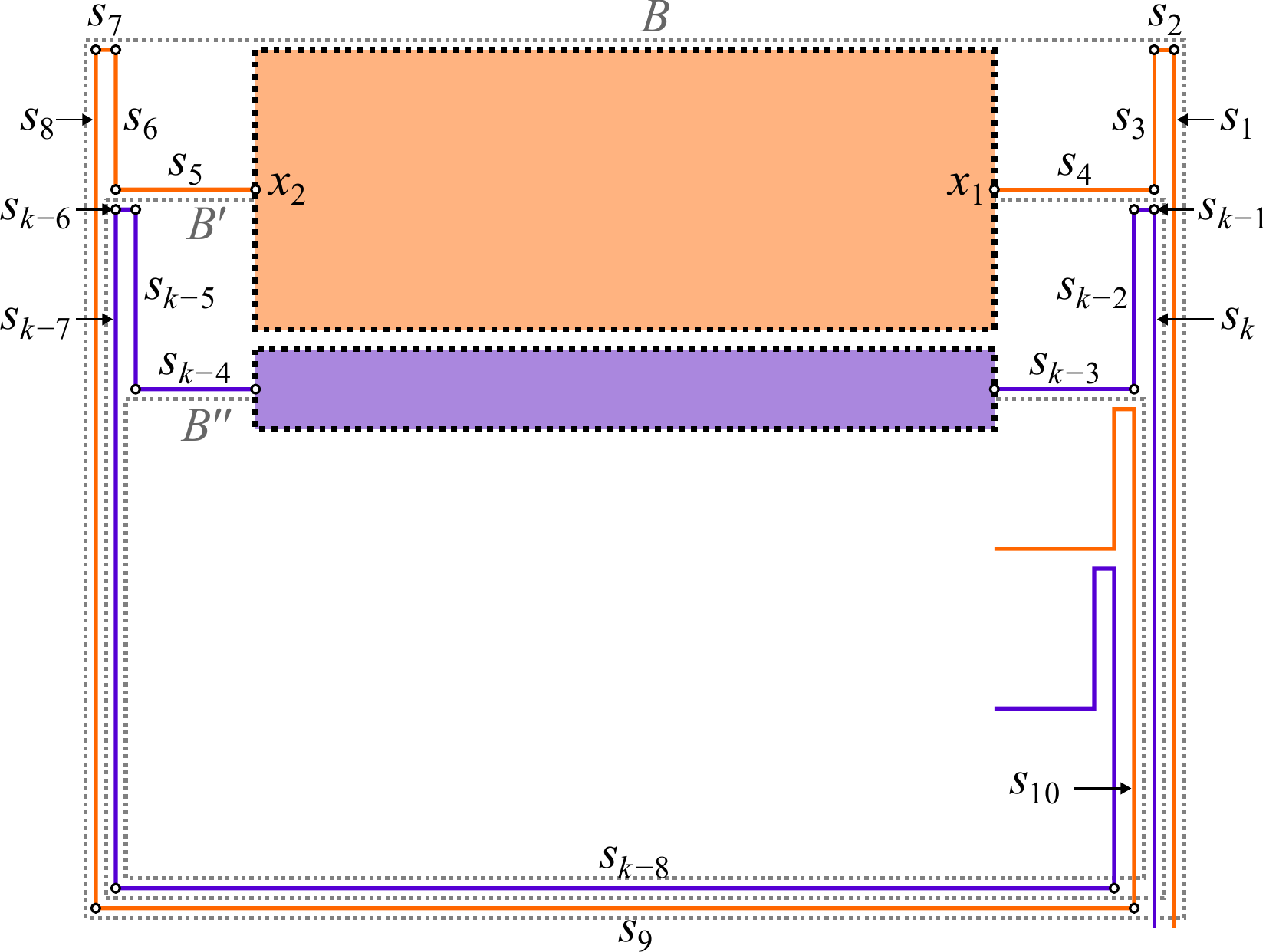}
    \caption{Forced folding of spiral gadget from Figure~\ref{fig:overview}.}
    \label{fig:spiral-analysis}
  \end{figure}

  We argue that segments $s_1, s_2, s_3, s_4$ have forced foldings.
  Because $s_1$ is maximally right in~$B$,
  segment $s_2$ of length $1$ must go left
  to avoid intersecting the right side of the frame.
  Because $s_1$ reaches the very top of~$B$,
  segment $s_3$ must then go down
  to avoid intersecting the top side of the frame.
  Because $s_2$ has length~$1$, so $s_3$ is just one grid line left of~$s_1$,
  segment $s_4$ must go left to avoid intersecting~$s_1$.
  If $s_1$ is in fact not followed by a U-turn, then we view segments $s_2$
  and $s_3$ as having length~$0$,
  and instead $s_4$ directly turns;
  in this case, $s_4$ must turn left to avoid intersecting the right side of
  the frame.

  Segment $s_4$ reaches one endpoint $x_1$ of the topmost clause row.
  We determine the location of the other endpoint~$x_2$,
  and the locations of the following segments $s_5, s_6, s_7, s_8, s_9$,
  as follows.
  Let $s_5$ denote the next segment of the spiral gadget,
  which is attached to~$x_2$.
  Segment $s_8$ is the entire height of the bounding box~$B$,
  which enforces the set of $y$ coordinates of its ends,
  but not whether it points up or down.
  In fact, $s_8$ must point down: otherwise, segment $s_9$ would be at the
  top and go rightward to be adjacent to (but not quite intersecting)~$s_2$,
  and then segment $s_{10}$ would be forced to go down
  (to avoid leaving $B$) and intersect~$s_4$:
  $s_{10}$ is at least as long as the height of the bottommost dotted box,
  which we assumed is the tallest dotted box, so is at least as long as~$s_3$,
  which is half the height of the current top dotted box.
  Thus $s_9$ must be at the very bottom of the bounding box~$B$.
  Segments $s_7$, $s_6$, and $s_5$ then have forced orientations
  of left, up, and left (from~$x_2$),
  to avoid local intersection with the left side of the frame,
  the top side of the frame, and $s_8$, respectively.

  At this point, we have shrunk the effective bounding box $B$
  by $1$ unit on the right, left, and bottom sides
  (because of segments $s_1$, $s_8$, and $s_9$)
  and put the top side just below segments $s_4$ and~$s_5$.
  The new bounding shape $B'$ is not exactly a box,
  because the topmost clause row has unknown shape,
  but we can treat it as a box because we will only argue about
  columns that do not overlap any clause/variable row.

  Next we look at the other end of the chain.
  We apply the $s_1, s_2, s_3, s_4$ argument to determine the foldings of
  $s_k, s_{k-1}, s_{k-2}, s_{k-3}$.
  Then we apply the $s_5, s_6, s_7, s_8, s_9$ argument to determine the
  foldings of $s_{k-4}, s_{k-5}, s_{k-6}, s_{k-7}, s_{k-8}$.
  Now we have shrunk the effective bounding box $B'$
  by $1$ unit on the right, left, and bottom sides,
  and put the top side just below segments $s_{k-3}$ and~$s_{k-4}$.

  With this smaller bounding box $B''$, we can repeat the above arguments,
  alternating between analyzing the next nine segments from the front
  and the previous nine segments from the back.
  Here we exploit the symmetry of the spiral gadget construction:
  each round with a smaller box looks just like the previous round.
  The only special segment is $s_{(k+1)/2}$ which connects the front and
  back of the chain directly instead of via a variable/clause row,
  but this only simplifies the argument.
  In the end, we determine all edges of the spiral gadget.
\end{proof}

\subsection{Insulation Gadget}
\label{sec:insulation}

\figurename~\ref{fig:insulation} shows an \defn{insulation gadget}
which occupies an entire row in \figurename~\ref{fig:overview}.
This gadget works on a half-grid relative to all other gadgets
except the frame; in other words, the spiral, variable, sheath, and clause
gadgets are all scaled $2 \times$ relative to the insulation gadget,
which is scaled $5 \times$ relative to the frame gadget,
for a total scale of $10 \times$
for the spiral, variable, sheath, and clause gadgets.
For simplicity, we allow segments of length $\frac 1 2$
so that the spiral remains on a $1 \times 1$ grid;
but in the end everything will be scaled by $10 \times$ to ensure
all lengths are integers divisible by~$5$.

\begin{figure}[ht]
  \centering
  \includegraphics[scale=0.666]{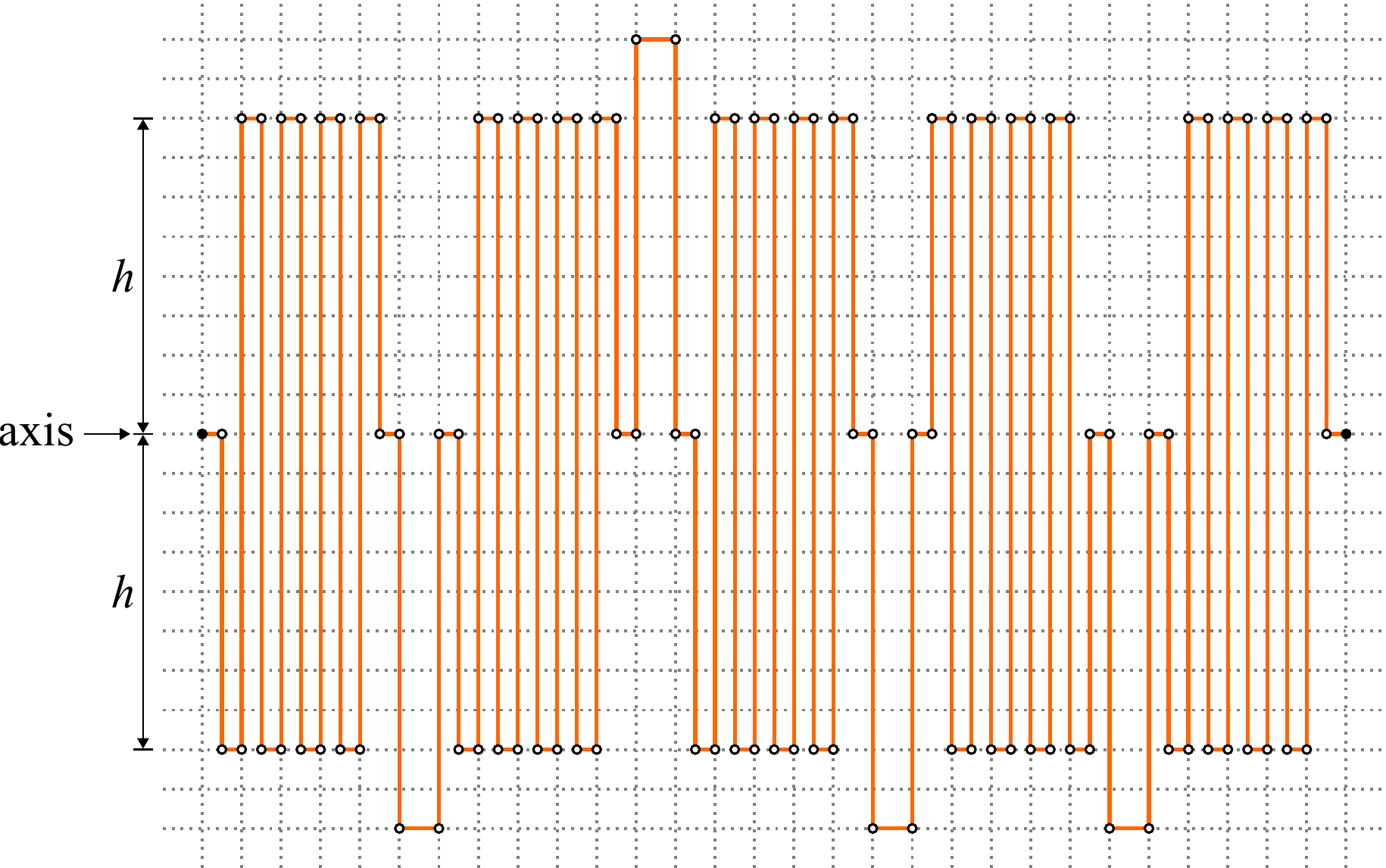}
  \caption{An example insulation gadget of height~$h$.
    (Straight vertices are not drawn to simplify the figure.)}
  \label{fig:insulation}
\end{figure}

The insulation gadget starts and ends with horizontal segments of
length $\frac 1 2$ that lie on a common horizontal \defn{axis}.
In between, the gadget consists of an alternation between
``grilles'' and ``tabs'', starting and ending with a grille,
separated by horizontal segments of length $\frac 1 2$
that also lie on the axis.
A \defn{grille} consists of a sequence of segments with lengths in the pattern
$h, \frac 1 2, 2 h, \big(\frac 1 2, 2 h\big)^{2 k}, \frac 1 2, h$,
where ${\cdots}^{2 k}$ denotes an even number of repetitions,
resulting in an integer width $k+1$ (for any desired integer $k \geq 0$).
A \defn{tab} consists of three segments with lengths $h+2, 1, h+2$,
which has width~$1$.
Notably, tabs lie on the integral grid used by all other gadgets
other than the frame, and extend $2$ units farther than grilles.
All grilles and tabs use the same integer parameter~$h$,
the \defn{half-height} of the insulation gadget.

To force the folding of an insulation gadget, we need the property that $h$
is larger than the height of any possible folding of the adjacent rows of
gadgets above or below the insulation gadget.
This property is easy to achieve by setting $h$ to be larger than
the sum of lengths of all vertical segments in those rows.

\begin{claim} \label{claim:insulation}
  In any folding of the entire construction,
  an insulation gadget must be folded as in \figurename~\ref{fig:insulation},
  with each grille and tab optionally reflected through the axis.
\end{claim}

\begin{proof}
  Because the endpoints of the insulation gadget (filled in black)
  are attached to the spiral gadget,
  their positions are fixed and the incident length-$\frac 1 2$
  segments must be horizontal by Claim~\ref{claim:spiral}.
  We refer to the extension of these two extreme segments as the \defn{axis}.

  Each grille and tab has an initial choice for its first segment
  (which has length~$h$) to fold up or down,
  corresponding to choosing one folding or its reflection through the axis.
  After this first segment,
  the chain is at $\pm h$ vertical distance from the axis.
  The third segment, which has length $h$ or $2 h$, must go toward the
  axis (i.e., in the opposite direction as the first segment);
  otherwise, the insulation gadget penetrates the adjacent row of gadgets
  by a vertical distance of at least~$h$, which prevents the chain in that row
  from connecting its two endpoints without collisions, by our assumption that
  $h$ is larger than the height of any possible folding of the row.
  Thus the second segment (which has length~$\frac 1 2$) must go rightward;
  otherwise, the third segment would intersect the previous horizontal segment.

  For the grille, any remaining segments of length $2 h$ must alternate in
  direction, always going toward the axis, by a symmetric argument,
  as we remain at $\pm h$ vertical distance from the axis;
  and all intervening segments of length $\frac 1 2$ must go rightward
  or they would cause collision between the previous and next vertical segments.

  For both the grille and tab, the last segment (which has length~$h$)
  must again go toward the axis, and thereby return to the axis.
  The next-to-last segment (of length $\frac 1 2$ for a grille and
  $1$ for a tab) must again go rightward: for a grille, going left would
  cause collision between the previous and last vertical segment,
  and for a tab, going left would cause collision between the previous grille
  (which exists because the insulation gadget starts with a grille)
  and the last vertical segment.

  Therefore every grille and tab is forced modulo the initial up/down choice.
\end{proof}

For understanding the impact of insulation on other gadgets,
we can restrict attention to the occupied points of the integral grid in
\figurename~\ref{fig:insulation}.
Grilles act as walls: independent of whether they are flipped through the axis,
they block the same $k+1 \times 2 h + 1$ rectangle of points.
Tabs act as local binary choices (wires):
they either block a top $2 \times h+2$ rectangle of points above the axis
and leave empty the bottom $2 \times h+2$ rectangle of points below the axis,
or vice versa.
In particular, we use the fact that the vertical gap between two consecutive
grilles (on the side that does not have the tab) is only two grid points wide.
Thus, if an integral chain entered and exited such a gap, it must have exactly
two turns separated by a segment of length~$1$; no more turns are possible.

Tabs in the insulation gadget provide communication wires between
gadgets above and below the insulation, while the forced blocking of grilles
give those adjacent gadgets an effective bounding box.
The insulation gadget can support any pattern of (integer-aligned) tabs
provided no two tabs are consecutive and there is no tab at the left or
right extreme: we simply fill in the remaining space with grilles.

\subsection{Choice Gadget}

\begin{figure}[th]\centering
  \begin{subfigure}{0.27\linewidth}\centering
    \includegraphics[scale=0.666]{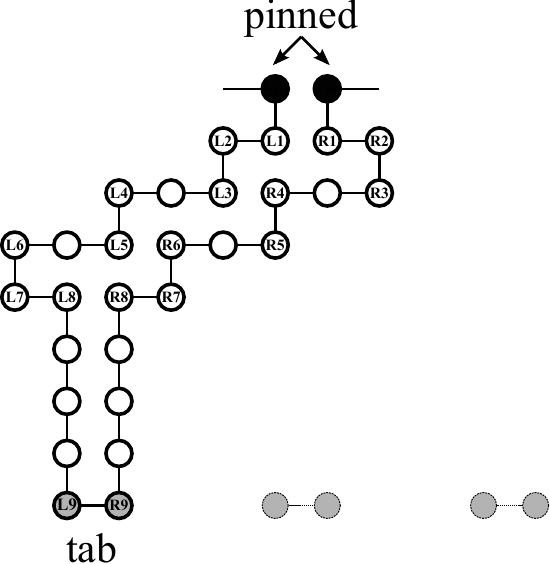}
    \caption{The gray tab can be placed in one of three places
      below the pinned vertices.}
    \label{fig:choice:tabs}
  \end{subfigure}\hfill
  \begin{subfigure}{0.575\linewidth}\centering
    \includegraphics[scale=0.666]{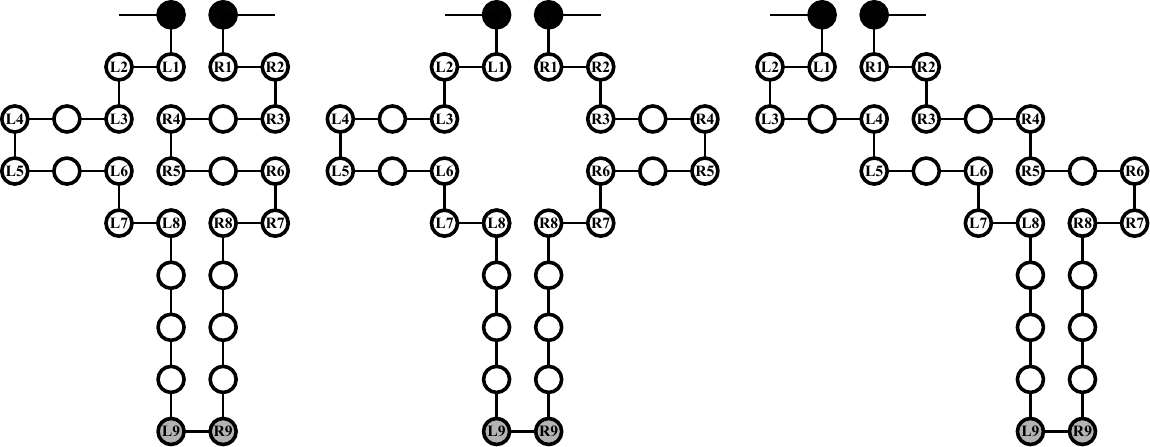}
    \caption{Representative configurations (modulo reflection).\newline\newline}
    \label{fig:choice3}
  \end{subfigure}\hfill
  \begin{subfigure}{0.1\linewidth}\centering
    \includegraphics[scale=0.666]{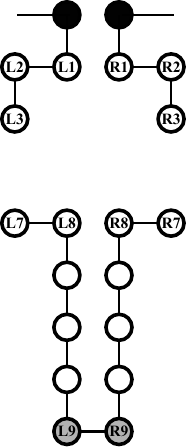}
    \caption{\raggedright Forced choices.\newline}
    \label{fig:choice-force}
  \end{subfigure}\hfill

  \medskip

  \begin{subfigure}{\linewidth}\centering
    \includegraphics[scale=0.666]{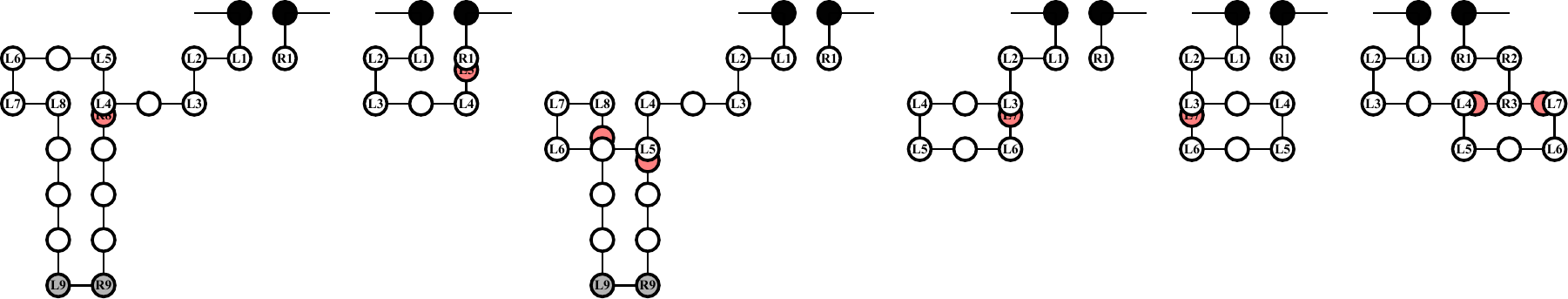}
    \caption{Attempting to turn L4 or L6 upward causes intersection.}
    \label{fig:choice-fail}
  \end{subfigure}
  \caption{Choice gadget, where black vertices are pinned.}
  \label{fig:choice}
\end{figure}

\figurename~\ref{fig:choice} illustrates the \defn{choice gadget},
which is the central part of a clause gadget.
We refer to the two gray vertices as the \defn{tab} of this gadget.
For now, we assume the two black endpoints are horizontally adjacent;
we will effectively pin these endpoints later.
We will also assume that the both ends of the chain turn the same direction ---
either upward or downward --- at these endpoints, which will be forced by
the long chains attached to the endpoints.
Under these assumptions, there are six types of configurations
as characterized by the tab locations:

\begin{claim} \label{claim:choice}
Assuming the endpoints are horizontally adjacent both turn the same
direction (upward or downward),
the tab of a choice gadget can be placed in exactly six locations
(with three different horizontal shifts).
\end{claim}

\begin{proof}
We claim that there are ten possible embeddings of the choice gadget:
the three in \figurename~\ref{fig:choice3},
their reflections through the vertical line bisecting the endpoints
(which adds two more, as the middle diagram is reflectionally symmetric),
and the reflections of these five embeddings through the horizontal line
connecting the endpoints.
We assume by symmetry that the endpoints both turn downward,
reducing to the first five embeddings.

To enumerate all embeddings,
we consider which way each $90^\circ$ vertex can turn.
As in Figure~\ref{fig:choice},
we label the $90^\circ$ vertices after the left endpoint by
L1, L2, \dots, L9, where L9 is the left tab vertex;
and symmetrically label the $90^\circ$ vertices before the right endpoint by
R1, R2, \dots, R9, where R9 is the right tab vertex.
These $90^\circ$ vertices offer a sequence of choices, alternating between
turning left vs.\ right and turning up vs.\ down.

Some of these choices are immediately forced;
refer to Figure~\ref{fig:choice-force}.
First, L1 and R1 (the two vertices adjacent to the endpoints)
must turn away from each other,
in order to not immediately intersect each other.
Second, L2 and R2 must turn downward,
in order to avoid intersecting the other neighbors of the endpoints.
Third, L9 and R9 (the tab vertices) must both point upward:
if L9 pointed downward, say, then L8 would be so much lower than L9
that the chain could not reach the left endpoint without intersection.
Fourth, L8 and R8 must turn away from each other,
in order to not immediately intersect each other.
Furthermore, L8 must point left and R8 must point right ---
otherwise, they could not reach their respective endpoints
without crossing --- so L7 must turn right and R7 must turn left.

\figurename~\ref{fig:choice3} and its reflections correspond to
always (in particular, L4, R4, L6, R6) choosing to turn down instead of up;
making all possible left vs.\ right choices for L3 and R3; and
making all possible left vs.\ right choices for L5 and R5
that keep the tab vertices in adjacent columns.
The remaining cases to consider are when L4, R4, L6, or R6 turn up.
By symmetry, it suffices to consider the cases when L4 or L6 turn up.
Figure~\ref{fig:choice-fail} shows that these cases intersect,
no matter what choices we made for L3 and L5.
\end{proof}

\subsection{Clause and Sheath Gadgets}

The clause gadget, shown in \figurename s~\ref{fig:hook1} and~\ref{fig:hook2},
consists of three separate chains (with endpoints marked black):
one extending the choice gadget,
and one ``sheath'' above and one below the choice chain.
On the left and right ends of the choice gadget,
we add suitably long horizontal segments.
Above and below the choice gadget, we add a \defn{sheath gadget} forming
a rectangular ``container'' around the choice gadget, and
up to five outward \defn{hooks} attached.
Each hook is a path of three segments (vertical, horizontal, vertical)
doubled to have thickness~$1$.

\begin{figure}\centering
  \includegraphics[scale=0.5]{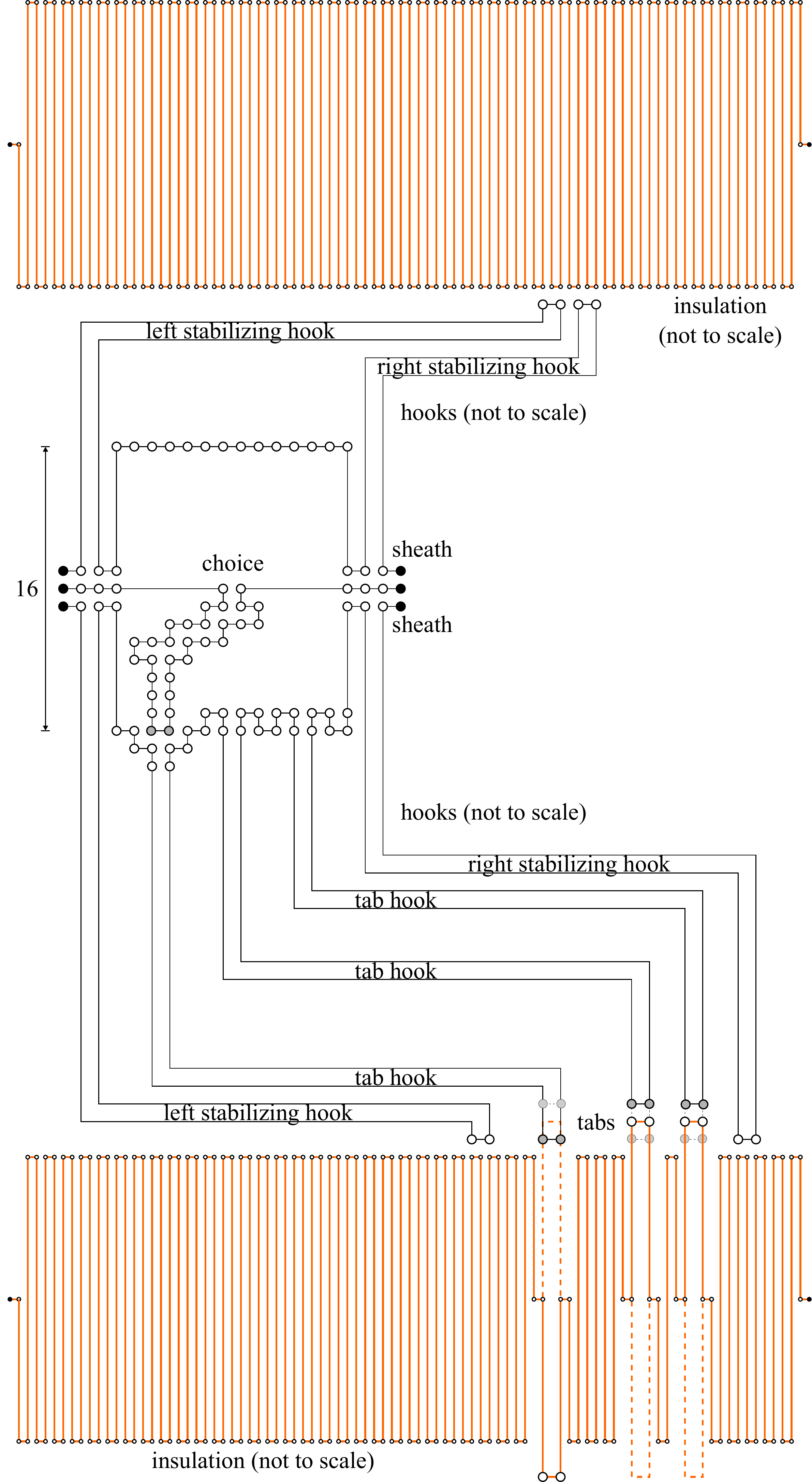}
  \caption{One version of clause gadget
    (three chains of black segments with black endpoints)
    and its interaction with the neighboring insulation gadgets
    (colored segments).
    This version has three tabs on the bottom,
    and the leftmost tab extended.
    Each tab shows the unused alternate state with dashed lines.
    (Some vertices are not drawn to simplify the figure.)}
  \label{fig:hook1}
\end{figure}

\begin{figure}\centering
  \includegraphics[scale=0.5]{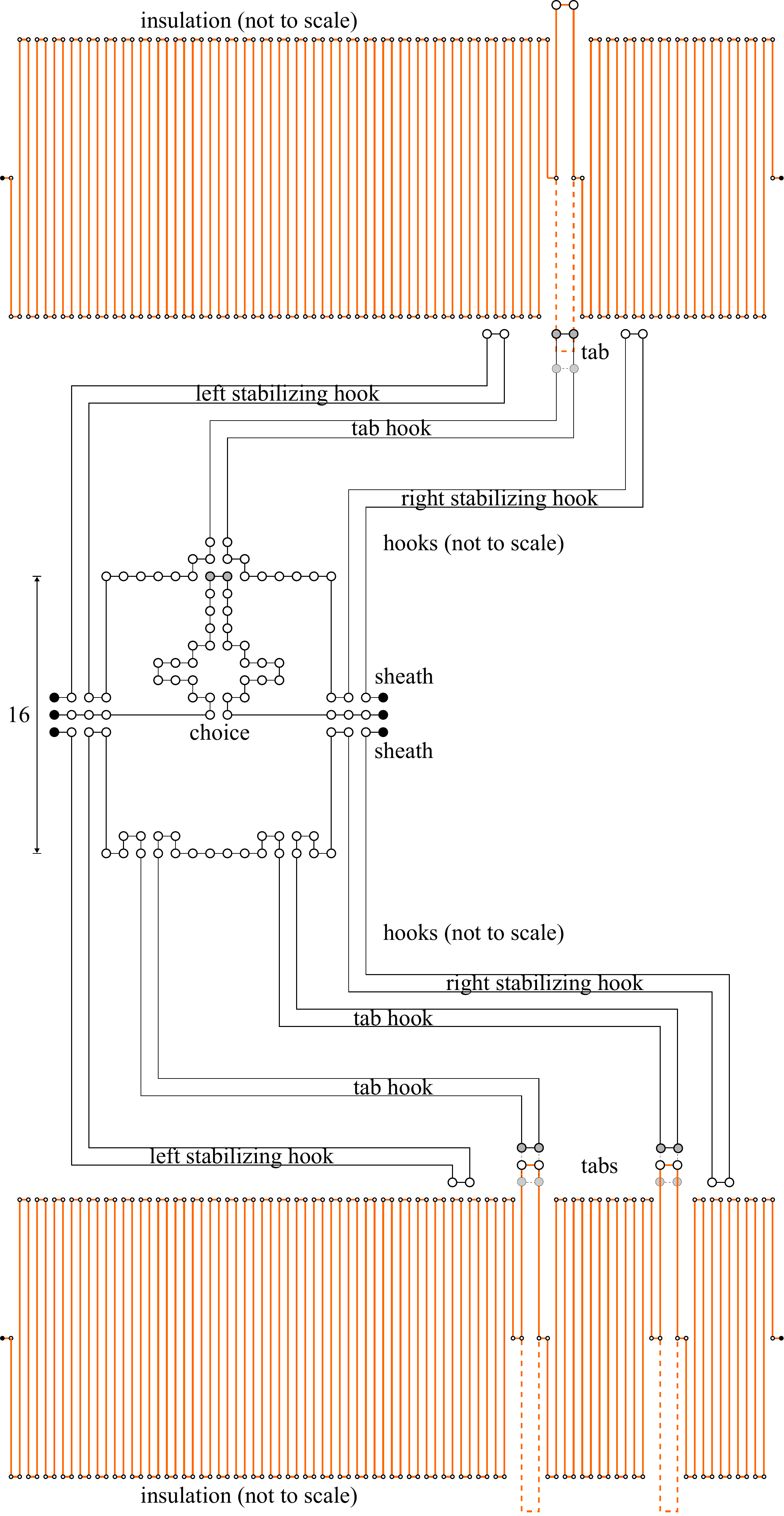}
  \caption{Another version of clause gadget
    (three chains of black segments with black endpoints)
    and its interaction with the neighboring insulation gadgets
    (colored segments).
    This version has two tabs on the bottom and one option on the top,
    and the upper tab extended.
    Each tab shows the unused alternate state with dashed lines.
    (Some vertices are not drawn to simplify the figure.)}
  \label{fig:hook2}
\end{figure}

At minimum, each sheath of the clause gadget has one \defn{stabilizing hook}
immediately left and right of the container,
and each of these hooks has its end adjacent to a grille of the
adjacent insulation gadget.
In addition, the two sheaths of a clause have a total of three \defn{tab hooks}
attached to the container,
for connections to variable gadgets either above or below the clause.
Tab hooks are attached to the sheath (container) via two length-$1$ segments
on either side, allowing for the hooks to extend or retract their \defn{tabs}
(gray vertices) vertically with an offset of~$2$; the attachment is aligned
so that the choice gadget's extended tab forces the corresponding sheath
tab hook to be extended.
The tabs of a tab hook are horizontally aligned with a tab of the
adjacent insulation gadget;
when the tab hook is retracted, the insulation tab is immediately adjacent
(but not intersecting), and when the tab hook is extended,
the insulation tab would intersect and so is forced to flip away.
The horizontal segments of different hooks are separated enough
from each other so that each hook can freely extend or retract.

\figurename s~\ref{fig:hook1} and~\ref{fig:hook2} show two different versions
of the clause gadget.  In general,
for each connection in $G_\phi$ from this clause to a variable gadget
in the row below, we assign one of the choice gadget's
three possible horizontal shifts of its tab to the bottom side,
and add a corresponding tab hook to the bottom sheath that routes the tab to be
horizontally aligned with the corresponding variable,
along with a corresponding tab to the insulation gadget below.
When the choice gadget places its tab in this location,
the tab hook must extend (shift down by~$2$),
forcing the insulation tab to flip down.
(This will enable communication with the variable below the insulation gadget.)
Similarly, for each connection in $G_\phi$ from this clause to
a variable gadget in the row above, we assign one of the choice gadget's
three possible horizontal shifts of its tab to the top side,
add a corresponding hook tab to the top sheath, and
add a corresponding tab to the insulation gadget above the clause.

In \figurename~\ref{fig:hook1}, the clause is connected to three variables
below, and we have extended the leftmost tab and hook;
while in \figurename~\ref{fig:hook2}, the clause is connected to one
variable above and two below, and we have extended the upper tab and hook.
These two variations and their reflections through a horizontal line
are all the cases we need for clauses, as every clause either has all
three connections on one side or splits its connections into a group of one
and a group of two.
(We could additionally re-assign the horizontal shifts of the choice gadget's
tabs between up vs.\ down, but this flexibility is unnecessary.)
In each case, the inner choice gadget has exactly three foldings:
among the six foldings from Claim~\ref{claim:choice},
three of them intersect a horizontal part of a sheath.
For example, in \figurename~\ref{fig:hook1}, the three upper options for the
tab of the choice are prevented by not crossing the upper horizontal sheath.

To force the folding of a clause gadget, we need to specify features of the
overall layout
(already implicit in \figurename s~\ref{fig:hook1} and~\ref{fig:hook2}).
Specifically, the clause gadgets all appear in the leftmost quarter of the
construction, and the variable gadgets all appear in the rightmost quarter.
Thus the hooks' nonunit horizontal segments have length more than half
the width of the construction.

We also constrain the design of hooks as follows.
First, we require the first and last extreme vertical segments
to be at least $50$ longer than the two other vertical segments,
so that the hook's initial vertical travel is significantly longer than
the second vertical travel.
Finally, we require that all vertical segments of a hook
have length at least $\ell_{\min} = \max\{50, 16 m+21\}$.
(Recall that $m$ is the number of clauses.)

\begin{claim} \label{claim:hook}
  Consider a downward hook gadget, with a fixed starting point
  in the left quarter of the construction, such that the intended folding
  descends to $1$ above a tab of the insulation gadget below, or lower.
  Then this hook gadget has a unique folding,
  provided no vertical segment can go $\ell_{\min}$ above the starting vertex.
\end{claim}

\begin{proof}
  Refer to \figurename~\ref{fig:hook-fail}, which enumerates all possible
  foldings of a downward hook gadget up to the first intersection with
  a hypothetical rectangular bounding box.
  For example, folding (a) considers when the initial vertical segment goes
  up instead of down, while all other foldings consider when it goes down;
  folding (b) considers when the second segment goes left instead of right;
  foldings (c--h) consider when the third segment goes up instead of down;
  foldings (i--j) consider when the fourth segment goes left instead of right;
  folding (k) considers when the fifth segment goes down instead of up;
  folding (l) considers when the sixth segment goes right instead of left;
  and folding (m) considers when the seventh segment goes down instead of up.

  In each case other than (n), we argue an impossibility as follows.
  Foldings (d), (f), (i), and (m) have local intersections among segments
  starting at two vertices of distance $1$ from the intersection
  (because every nonunit segment of the hook gadget has another segment
  of length $\pm 1$ longer).
  Foldings (a), (c), and (e) go at least $\ell_{\min}$
  above the starting vertex,
  because every vertical segment of the hook gadget
  is at least $\ell_{\min}$ long.
  Folding (b) crosses the left side of the frame,
  because the starting vertex is in the left quarter of the construction,
  while the second segment of the hook gadget is longer than half the
  width of the construction.
  Foldings (h) and (l) cross the right side of the frame,
  because the second segment going right (as well as the third)
  puts us in the right quarter of the construction,
  while the sixth segment of the hook gadget is longer than
  half the width of the construction.
  Folding (g) descends at least $50$ below the intended bottom of the hook
  gadget (because the first and last segments are at least $50$ longer
  than the other vertical segments),
  while foldings (j) and (k) descend at least $\ell_{\min} \geq 50$ below
  (because the fifth segment has length at least $\ell_{\min}$).
  In any of these three cases, the hook gadget penetrates the
  insulation gadget below by at least~$50$.
  If it is horizontally aligned with a grille, we get immediate intersection.
  If it is horizontally aligned with a tab, we get crossings within the
  next three segments after the hook:
  (1)~if the hook is a tab hook, then the next two segments have length~$1$,
  and the turn immediately after crosses into a horizontally
  adjacent grille;
  (2)~if the hook is a left stabilizing hook, then the next segments have
  length $1$ and $7$ respectively, so again the following turn immediately
  crosses into a horizontally adjacent grille; and
  (3)~if the hook is a right stabilizing hook, then the next segment
  has length $> 1$, so it immediately crosses into a horizontally
  adjacent grille.
  Thus the only remaining folding is the intended folding (n).
\end{proof}

\begin{figure}\centering
  \includegraphics[scale=0.3]{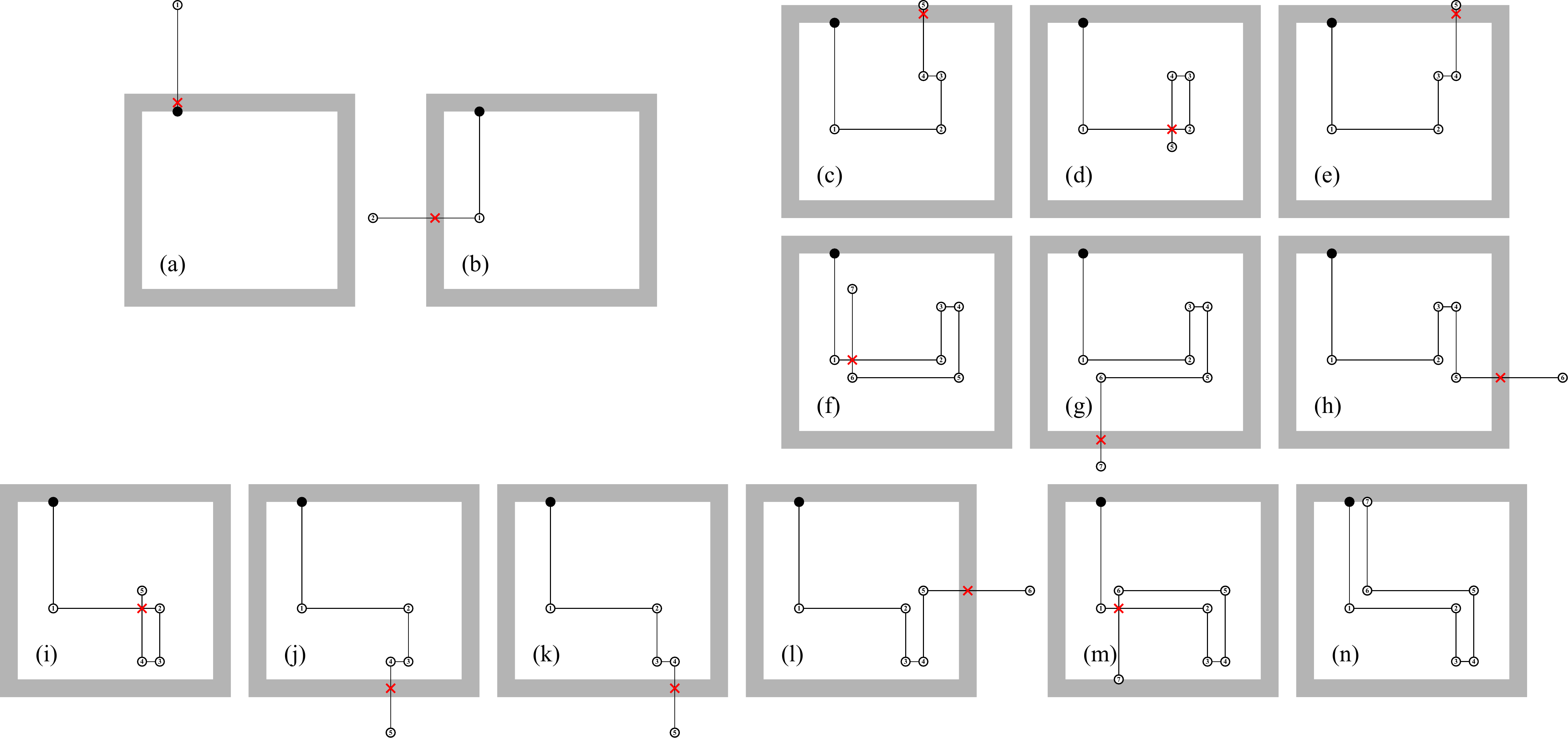}
  \caption{A hook gadget can fold in only one way (n)
    if it is constrained to lie within a rectangle (gray).}
  \label{fig:hook-fail}
\end{figure}

We argue that the clause gadgets must be folded as intended
in \figurename s~\ref{fig:hook1} and \ref{fig:hook2},
from left to right within the three rows of gadgets they occupy.
The leftmost endpoints of the rows are forced by
Claim~\ref{claim:spiral}, and the next claim allows us
to induct through all the clause gadgets.

\begin{claim} \label{claim:clause}
  Assume the left three black vertices of a clause gadget are placed as
  in the intended foldings
  (shown in \figurename s~\ref{fig:hook1} and \ref{fig:hook2})
  and the three chains start rightward.
  Then any folding must place the right three black vertices
  as in the intended folding; and at least one tab hook
  must be extended into the corresponding tab of the insulation gadget.
\end{claim}

\begin{proof}
  Define the $y=0$ line to be the horizontal line through
  the left endpoint of the clause gadget's choice chain, or equivalently,
  the left endpoint of the entire row's choice chain.

  First we show that no vertical segment of a bottom hook gadget
  can go up to $y = 16 m$, which is useful for applying Claim~\ref{claim:hook}.
  If it did, then the entire row's choice chain would be unable to connect
  its two endpoints (where the choice chain connects to the spiral):
  the total length of its vertical segments is exactly $16 m$,
  so it can reach only $16 m$ above the origin,
  so it could never reach above the wayward vertical segment.

  Consider the bottom sheath from left to right.
  The left stabilizing hook starts at $y=-1$,
  so Claim~\ref{claim:hook} guarantees the correct folding,
  because no vertical segment can go up to
  $y=-1+\ell_{\min} \geq 16 m$.
  The next segment (which has length $1$)
  must go right to avoid intersecting the first segment of the sheath.
  The next segment (which has length $7$) must go down, to $y=-8$,
  to avoid intersecting the choice chain.
  The next segment must go right to avoid intersecting the first hook.

  Now consider the zero, one, two, or three tab hooks on this sheath,
  one at a time from left to right.
  Each tab hook is prefixed by a vertical segment of length~$1$,
  which can go up or down according to whether the second hook
  is retracted or extended, and a horizontal segment of length~$1$,
  which we will argue must go right.
  The first tab hook starts at $y=-7$ or $y=-9$,
  so Claim~\ref{claim:hook} guarantees the correct folding,
  because no vertical segment can go up to $y=-9+\ell_{\min} \geq 16 m$.
  This correct folding implies that the prefix horizontal segment
  goes right; if it instead went left, then its right endpoint would be
  where the tab hook gadget ends, causing an intersection.
  The tab hook is suffixed by a horizontal segment of length~$1$,
  which must go right to avoid immediately intersecting the hook,
  and a vertical segment of length $1$, which we will argue must go
  the opposite direction of the prefix vertical segment.
  But even if both vertical segments surrounding tab hooks all go down,
  we reach $y=-10$ after the first tab hook, $y=-12$ after the second tab
  hook, and $y=-14$ after the third tab hook (if they exist).
  All of the tab hooks start at $y \geq -12$,
  so Claim~\ref{claim:hook} guarantees the correct folding,
  because no vertical segment can go up to $y=-12+\ell_{\min} \geq 16 m$.

  After the hooks, we have a horizontal segment,
  which must go right to avoid intersecting the previous hook
  (the left stabilizing hook if there are no bottom tab hooks);
  a vertical segment of length~$7$, which we will argue goes up;
  and a horizontal segment of length~$1$, which we will argue goes right.
  Thus we must be at $y \geq -21$ when we reach the right stabilizing hook.
  Claim~\ref{claim:hook} guarantees the correct folding,
  because no vertical segment can go up to
  $y = -21 + \ell_{\min} \geq 16 m$.
  This correct folding implies that the preceding horizontal segment of length
  $1$ goes right; if it instead went left, then its right endpoint would be
  where the stabilizing hook gadget ends, causing an intersection.
  Finally we have a horizontal segment of length~$1$,
  which must go right to avoid immediately intersecting the
  right stabilizing hook.

  At this point, we have guaranteed that all bottom hook gadgets fold
  correctly, and all other horizontal segments of the sheath go right.
  Thus we have determined the $x$ coordinates of the entire bottom sheath
  to be as in the intended foldings of
  \figurename s~\ref{fig:hook1} and \ref{fig:hook2}.
  In particular, the tip of the right stabilizing hook is
  (as in the intended folding)
  horizontally aligned with a grille of the insulation gadget below,
  so to avoid collision the $y$ coordinate must be
  strictly above the insulation gadget, i.e.,
  no lower than in the intended folding.
  This guarantees that the starting $y$ coordinate for the right stabilizing
  hook is no lower than in the intended folding, i.e., at $y \geq -1$.
  By applying the same argument to the top sheath, its final endpoint is at
  $y \leq 1$.
  To leave room for the choice chain to exit on the right
  (necessary to reach the right endpoint of its row where it meets the spiral),
  the bottom and top sheaths must in fact end at exactly $y=-1$ and $y=1$,
  respectively.

  Now we analyze the starting $y$ coordinate for the bottom tab hooks.
  We have already argued that the leftmost tab hook
  starts at $y=-7$ (retracted) or $y=-9$ (extended).
  Now that we have fixed the folding of the right stabilizing hook,
  we determine the vertical segment of length $7$ that precedes it;
  in particular, it must go up to reach $y=-1$.
  Thus the rightmost tab hook must start at $y=-7$ or $y=-9$.
  The remaining case is when there are three bottom tab hooks.
  If the leftmost and rightmost tab hooks are both retracted ($y=-7$)
  or extended ($y=-9$), then the middle tab hook could reach
  $y=-5$ or $y=-11$, respectively.
  The first case $y=-5$ is impossible, because then all three tab hooks
  are retracted, so the choice gadget must intersect by
  Claim~\ref{claim:choice}.
  The second case $y=-11$ acts the same as $y=-9$:
  it is still extended into the corresponding tab of the insulation gadget
  below.

  A symmetric argument determines the folding of the top sheath.

  Finally we analyze the choice chain.
  Because there is only a single $y$ coordinate of space between
  the top and bottom sheaths on the right, namely $y=0$,
  the final horizontal segment of length $9$ must be at $y=0$.
  To avoid intersecting the first horizontal segment of length~$9$,
  the $x$ coordinate of the right endpoint of the choice chain
  must be at or right of where it is in the intended folding
  ($19$~right of the left endpoint).
  Now consider the entire row's choice chain,
  whose overall width is determined by the spiral gadget.
  The overall width is the sum of the widths of the individual clauses'
  choice chains.
  If any individual clause's choice chain were longer than
  the $20$ points used by the intended folding,
  then by conservation of the sum, some other clause's
  choice chain would have to be shorter, which we have argued is impossible.
  Therefore every clause's choice chain has a width of exactly $20$ points.
  Thus the endpoints of the choice gadget must be horizontally adjacent
  as in the intended folding, so Claim~\ref{claim:choice} guarantees that
  at least one tab is extended.
\end{proof}

\subsection{Variable Gadget}

\figurename~\ref{fig:variable} illustrates the \defn{variable gadget}.
For a variable $v$ that appears in $k$ clauses
(as positive literal $v$ or negative literal $\neg v$),
the variable gadget consists of two \defn{zig-zag paths}
occupying a (dotted) rectangle of points of width $3k+1$ and height~$3$.
The two zig-zag paths are joined by a vertical length-$3$ \defn{cap},
followed by a horizontal \defn{baseline} of length $3k+4$,
followed by another vertical length-$3$ cap.
In the intended folding, the baseline is in the middle of the available space
(aligned with the black endpoints), separating the upper and lower zig-zag
paths, which are folded to look identical (so measured along the chain,
they are reversals of each other).
At the beginning and end of the chain, we add a \defn{bookend}
consisting of vertical segments of length $3$, $6$, and $4$;
horizontal segments of length $2$ in between the vertical segments;
a horizontal segment of length $2$ incident to the black endpoint;
and a horizontal segment of length $3$ on the other end of the bookend.
The two intended solutions of the entire variable gadget are the
one shown in \figurename~\ref{fig:variable}
(corresponding to setting the variable to true)
and its reflection through the baseline
(corresponding to setting the variable to false).

\begin{figure}[th]\centering
\includegraphics[scale=0.5]{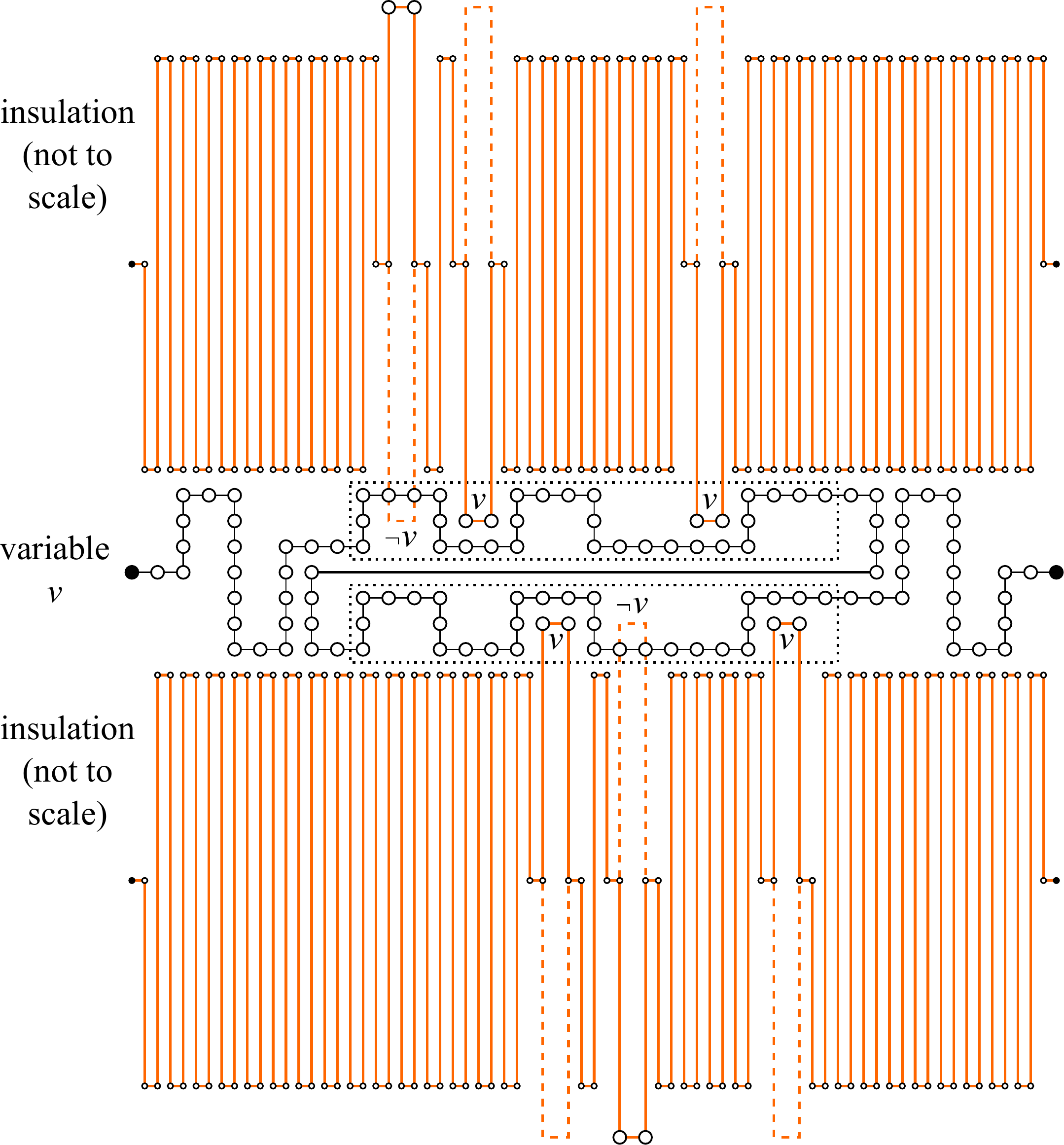}
\caption{A variable gadget for a variable $v$ that appears five times as
  $v$, $\neg v$, $\neg v$, $v$, $\neg v$, and~$\neg v$
  coming from tabs above, above, below, below, above, and below, respectively.
  The dotted areas outline the two zig-zag paths.
  (Some straight vertices are not drawn to simplify the figure.)}
\label{fig:variable}
\end{figure}

Both zig-zag paths contain a horizontal segment of length $3$
(connecting four vertices) for each appearance of the variable.
The heights of the segments on the upper and lower zig-zag paths,
measured from the baseline,
are either $3$ and $-1$ respectively, or $1$ and $-3$ respectively.
Which option depends on whether the corresponding literal
uses the variable in its positive or negative form,
and on whether the clause that uses the literal is in a row
above or below this one.
The heights are $(1,-3)$ if and only if either the literal is $v$
and the clause is above,
or the literal is $\neg v$ and the clause is below;
in \figurename~\ref{fig:variable},
these are the second, fourth, and fifth pairs of horizontal segments.
In the other cases, the heights are $(3,-1)$;
in \figurename~\ref{fig:variable},
these are the first, third, and sixth pairs of horizontal segments.
We add vertical segments of length $2$ to transition when necessary
between heights $1$ and $3$ on the upper zig-zag path,
and between heights $-1$ and $-3$ on the lower zig-zag path.
Notably, we do not add vertical segments between length-$3$ horizontal
segments of the matching height;
in \figurename~\ref{fig:variable},
this occurs between the fourth and fifth pairs of horizontal segments.

For each occurrence of the variable in a clause in a row above,
we place a corresponding tab of the insulation gadget immediately above
the variable gadget.  This tab is horizontally aligned with the middle two
vertices of the length-$3$ horizontal segment of the upper zig-zag path.
When the tab is down (corresponding to a clause choosing this variable
to satisfy it), it comes down to height~$2$, which forces the variable to
flip so that this horizontal segment has height~$1$.
Similarly, for each occurrence of the variable in a clause in a row below,
we place a corresponding tab of the insulation gadget immediately below
that is horizontally aligned with the middle two vertices of the
length-$3$ horizontal segment, and which comes up to height~$-2$ when up.
This vertical alignment of the insulation gadgets places the grilles of the
insulation gadgets above and below at heights $4$ and $-4$ respectively.

\begin{claim} \label{claim:variable}
  For each variable gadget in a row,
  the only valid foldings are the one in \figurename~\ref{fig:variable}
  and its reflection through the baseline.
\end{claim}

\begin{proof}
  First we analyze the five-segment bookend at the beginning and end
  of each variable gadget's chain.
  None of the vertical segments could stick into
  the width-$2$ gap of a tab in an insulation gadget, because then the
  incident horizontal segment of length at least $2$ would intersect a grille.
  Thus these vertical segments lie within the height range $\pm 3$
  (where height $0$ denotes the middle of the space between insulation grilles),
  so the middle vertical segment of length $6$ occupies the full height
  range, and the first and last vertical segments (of length~$3$)
  must go in opposite directions as the middle vertical segment.
  In particular, the endpoints of the variable gadget must be at height~$0$.
  Furthermore, the four horizontal segments of a bookend
  must all go the same direction to avoid local intersection
  among the vertical segments, and this direction must be right
  to enable eventual connection from the left endpoint of the row
  to the right endpoint of the row (otherwise the length-$6$ vertical
  segment would cut them off).

  Now we make a global argument about the variable gadgets in the row,
  and select a specific variable gadget to consider.
  Define the \defn{width} of a variable gadget to be the signed
  horizontal distance between its endpoints,
  which must be positive because of the length-$6$ vertical segments
  serving as left-to-right barriers.
  The total width of the variable gadgets is fixed by the spiral gadget
  and Claim~\ref{claim:spiral}.
  If any variable gadget had width greater than the width of the intended
  foldings, then the width of some other variable gadget would be less than
  intended.
  Consider a variable gadget whose width is less or equal to
  the intended width.

  We claim that the left endpoint of the baseline is at least $1$ unit right of
  the last vertical segment (of length~$4$) of the left bookend.
  First, it must be right of the middle vertical segment (of length~$6$)
  of the left bookend, so there are only two intervening $x$ coordinates
  to consider.
  In either case, the baseline cannot have height within the height range
  of the last vertical segment (of length~$4$) of the left bookend.
  This leaves just two possible heights: $\pm 2$ and $\pm 3$
  (where $\pm$ depends on whether the left bookend is reflected).
  Assume by symmetry that the left bookend is folded as in
  \figurename~\ref{fig:variable}, so that the baseline has
  height $2$ or $3$ if not~$0$.
  We use that the baseline has, on both endpoints and hence on the left
  endpoint, an incident cap segment of length~$3$.
  This cap segment must go down from the baseline:
  going up would intersect a grille, or enter a tab gap which would then
  cause an intersection because the next horizontal segment has length
  at least~$2$.
  To avoid intersecting the final vertical segment of the left bookend,
  this forces the left endpoint of the baseline to be in the column
  in between the middle segment (of length~$6$) and final segment
  (of length~$4$) of the left bookend.
  This folding wedges the cap segment
  in between the middle and last vertical segments of the left bookend,
  so the next horizontal edge intersects one of those vertical segments.

  By a symmetric argument, the right endpoint of the baseline
  is at least $1$ unit left of the first vertical segment (of length~$4$)
  of the right bookend.
  These two bounds on the horizontal location of the baseline's endpoints
  contradict each other if the variable gadget has width smaller than intended.
  Therefore this variable gadget, and all variable gadgets,
  have the intended width.
  Furthermore, by the bounds on the endpoints,
  the baseline's horizontal location is exactly as in the intended folding.

  Next we argue that the baseline has an even height: $-2$, $0$, or $2$.
  We have already fixed the height of the right endpoint of the
  left bookend to be~$1$, which is odd.
  The vertical segments of the intervening zig-zag path are all length~$2$,
  which is even.  The only odd-length vertical segment between the left bookend
  and the baseline is the length-$3$ segment incident to the baseline.
  Thus the height of the baseline is even.

  Now we show that the baseline cannot have height $\pm 2$,
  leaving only height~$0$.
  If the baseline had height~$-2$, then the length-$3$ cap segment incident
  to the left endpoint must go up (to avoid the insulation below,
  as argued earlier) to height~$1$.
  Similarly, if the baseline had height~$2$, then the length-$3$ cap segment
  incident to the left endpoint must go down to height~$-1$.
  In either case, the cap segment intersects
  the final horizontal segment of the left bookend, which has height~$\pm 1$.

  At this point, we have determined the horizontal and vertical location
  of the baseline to be as in the intended foldings.
  We have also determined the foldings of the left and right bookends,
  up to reflection through the baseline.
  The baseline and bookends thus decompose the available space into two
  disjoint regions, roughly above and below the baseline.
  The first zig-zag path (attached to the left bookend)
  must start by going right (to avoid intersecting the middle vertical segment
  of the left bookend), then go up (to avoid intersecting the baseline),
  then go right (to avoid cutting off connectivity to the right endpoint),
  then go down (to avoid intersecting the insulation, which would cause
  intersection because the horizontal segments all have length at least~$4$),
  and so on.
  The second zig-zag path is similarly forced,
  which forces the right bookend to be reflected opposite from the
  left bookend (as in \figurename~\ref{fig:variable}).
  Therefore we have determined the entire folding of the variable gadget,
  up to the reflection of the left bookend, which reflects the entire
  folding through the baseline.
\end{proof}

One issue can arise when connecting clause gadgets to variable gadgets:
when multiple clauses connect via tab hooks to the same side of a variable
gadget, they also need to place two stabilizing hooks against a grille
at the transition point.  Given that the horizontal alignment of insulation
tabs is controlled by the variable gadget, we need to leave horizontal room
for such stabilizing hooks.  We can do so by adding a ``null'' occurrence of
the variable that is not used by any clauses, and has no corresponding
stabilizing hook; instead, it is surrounded by grilles.
These null occurrences can have zig-zag path heights of
$3$ and $-1$ respectively, or $1$ and $-3$ respectively; it does not matter.
For example, in Figure~\ref{fig:sample},
this modification occurs at the bottom of $v_3$ and
at the top of~$v'_3$.

\subsection{Putting Gadgets Together}

\figurename~\ref{fig:sample} shows how all the gadgets fit together for
an example instance.
We join together all upper halves of hook gadgets for $c_1,c_2,\ldots,c_m$;
all clause gadgets (and their flaps) for $c_m,c_{m-1},\ldots,c_1$;
all lower halves of the hook gadgets for $c_1,c_2,\ldots,c_m$;
and all variable gadgets for $v_1,v_2,\ldots,v_n$, in these orders.
Finally, we apply the frame gadget to this closed chain and the intended
bounding box,
with an edge on a path joining the upper halves of the hook gadgets,
or an edge on a path joining the variable gadgets. 

Here we claim that we always have such an edge visible from the outside.
When the Hamiltonian cycle $\kappa$ has an edge visible from the outside, we can use any one of them.
(In \figurename~\ref{fig:sat}, the edges $\{c_1,c_2\}$ and $\{v_1,v_2\}$ are visible.)
Otherwise, we have a clause vertex $c_i \in C$ on the outer boundary. 
Then we can take an edge from the upper half of the hook gadget for the clause $c_i$.
This completes the construction.

\subsection{Correctness}

Now we conclude the proof of Theorem~\ref{th:loop}.
This reduction can be done in time polynomial in the size of~$\phi$:
we apply the Pilz reduction to draw $\phi$ in the spiral fashion
with connections only between adjacent rows, push the clauses to the
left of each row and push the variables to the right of each row,
and route the hook connections in a planar fashion ---
with the horizontal part of all hooks in a row stacking up vertically,
while leaving enough vertical space between them to allow for hooks to
freely extend or retract.
It remains to show that an instance $\phi$ of
planar 3SAT is satisfiable if and only if the resulting
fixed-angle orthogonal equilateral closed chain has a planar embedding.

When the linked planar 3SAT instance is satisfiable,
at least one literal of each clause is satisfied by the assignment.
The clause gadget then chooses the corresponding tab location for the
choice gadget, and extends the corresponding hook gadget,
while retracting the other tabs.
The adjacent insulation gadget then flips its corresponding tab
to avoid overlap.
Reflecting the variable gadget into the assignment corresponding to the
literal means that this insulation tab will not intersect the
variable gadget.
This folding avoids all intersections because the assignment is satisfied.

On the other hand, when the loop has an embedding,
the frame gadget folds into the intended bounding box
by Claim~\ref{claim:frame-contain},
and the spiral gadget folds as intended by Claim~\ref{claim:spiral}.
Each insulation gadget folds as intended by Claim~\ref{claim:insulation}.
Each row of clause gadgets folds as intended by
Claim~\ref{claim:clause}, and each row of variable gadgets
folds as intended by Claim~\ref{claim:variable}.
In particular, by Claim~\ref{claim:clause}, at least one tab hook
from each clause must be extended, which forces the corresponding
insulation tab to be folded into the corresponding variable gadget,
which forces the variable to have the satisfying assignment.
Therefore, the instance of planar 3SAT is satisfiable.

\section{HP Optimal Folding a Fixed-Angle Orthogonal Equilateral Open Chain is Strongly NP-complete}

We now turn to orthogonal equilateral open chains in the HP model,
where the vertices are bicolored $\Red$ or $\Blue$,
and we wish to find a noncrossing configuration in 2D
that maximizes the number of H--H contacts.
In this section, we prove that this problem is NP-complete, despite the chain being open:
\begin{theorem}\label{th:chain}
HP optimal folding of a bicolored fixed-angle orthogonal
equilateral open chain is strongly NP-complete,
even if the chain has just two $\Red$ vertices.
\end{theorem}

\begin{proof}
We use the same reduction in the proof of Theorem~\ref{th:loop},
except for the frame gadget, which we replace with \figurename~\ref{fig:frame}.
The differences are that the bottom doubled segment extends very far
to the left --- more than 10 times the total length $L$
of the given scaled chain~$5 C$ --- and the chain is no longer
closed at the left end of the bottom doubled segment.
The leftmost two vertices of the bottom doubled segment
(the endpoints of the chain) are~$\Red$,
while all other vertices in the chain are~$\Blue$.

\begin{figure}[ht]\centering
\includegraphics[width=0.9\linewidth]{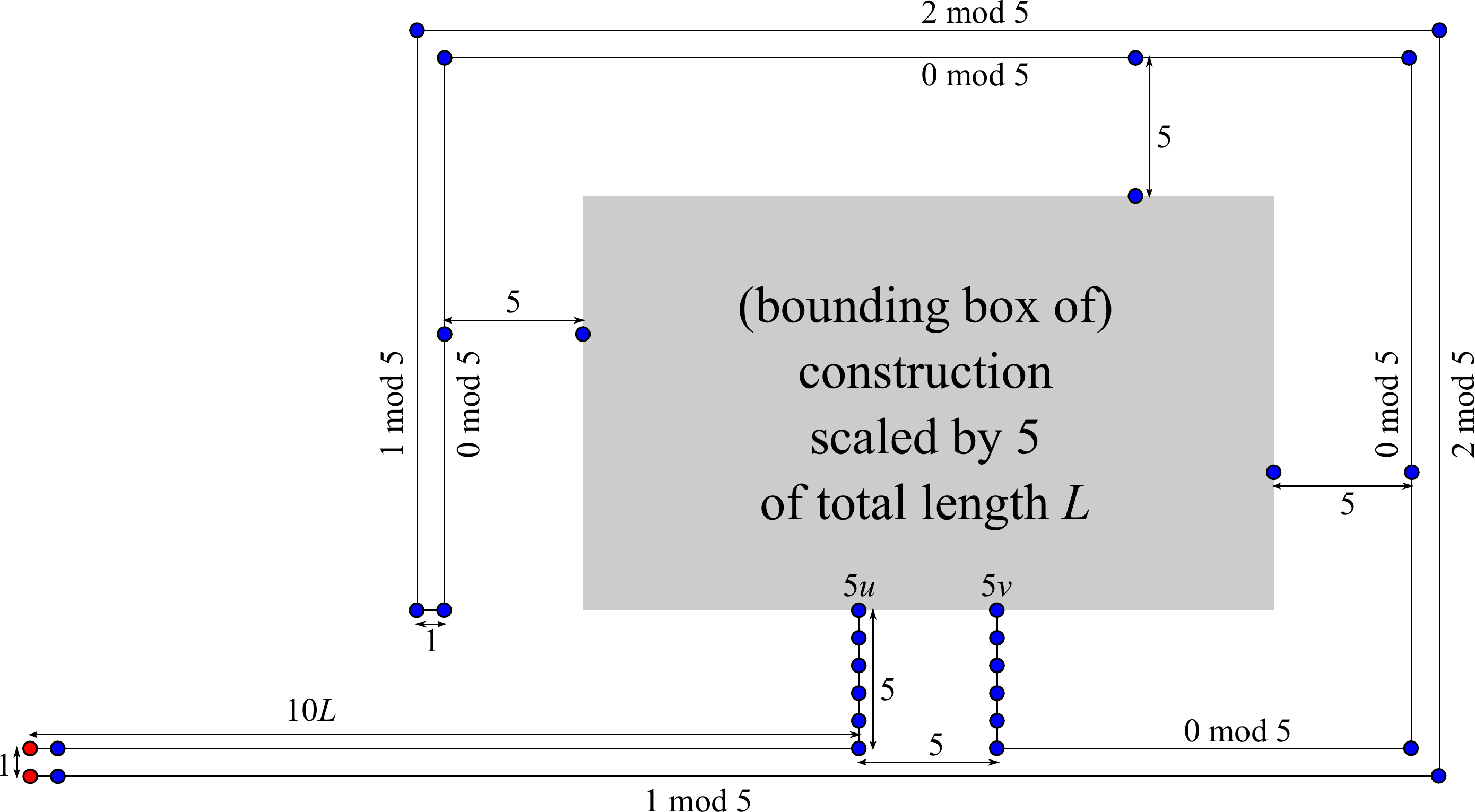}
\caption{A frame gadget for an HP chain.
  The two $\Red$ vertices are drawn red at the far left.}
\label{fig:frame}
\end{figure}

This reduction can be done in polynomial time.
Thus it suffices to show that this arrangement of the frame is the only way
to obtain the H--H contact at the two $\Red$ vertices.
Because the total length of the given construction inside of the frame
is at most~$L$, and the length of each segment of the frame is therefore
at most~$L$,
the total length of the chain except the bottom doubled segment
is at most $9 L$.
Hence to make the H--H contact between the two $\Red$ vertices,
the two long segments attached to the $\Red$ vertices must be arranged
in parallel as shown in \figurename~\ref{fig:frame}:
if the two long segments went in opposite directions, or were perpendicular
to each other, then the rest of the chain would not be long enough
to connect their ends together.
Thus the frame construction is forced to act like the closed chain of
\figurename~\ref{fig:frame-loop}, so it has a unique folding up to isometry
by the rest of the proof of Lemma~\ref{claim:frame-loop}.
\end{proof}

\section{Packing Fixed-Angle Orthogonal Equilateral Open Chains into Squares is Strongly NP-complete}

We now address some of the open questions from \cite{AbelDemaineDemaineEisenstatLynchSchardl2013}.
First, the authors ask whether a fixed-angle orthogonal equilateral open chain
(or in their terminology, an S--T sequence of squares, where each S square
must continue straight and each T square must turn left or right)
can be packed into a 2D square.
Second, they ask whether the problem remains hard when the chain occupies
a small fraction of the volume of the target shape.
(They ask this question for the 3D version of the problem, but it naturally
extends to the 2D version we consider.)
We answer both questions by showing that packing a fixed-angle orthogonal
equilateral open chain of length $O(s)$ into an $s \times s$ square is
strongly NP-complete.
This result is tight up to constant factors: if the chain has length $< s$, then
it can be packed into an $s \times s$ square via Observation~\ref{obs:chain}.

\begin{theorem}\label{th:square}
  Embedding a given fixed-angle orthogonal equilateral open chain into an
  $s \times s$ square is strongly NP-complete,
  even if the chain has length $O(s)$.
\end{theorem}

\begin{proof}
We use the same reduction in the proof of Theorem~\ref{th:chain},
except for the frame gadget, which we replace with \figurename~\ref{fig:frame_square}.

\begin{figure}[ht]
    \centering 
    \includegraphics[width=0.9\linewidth]{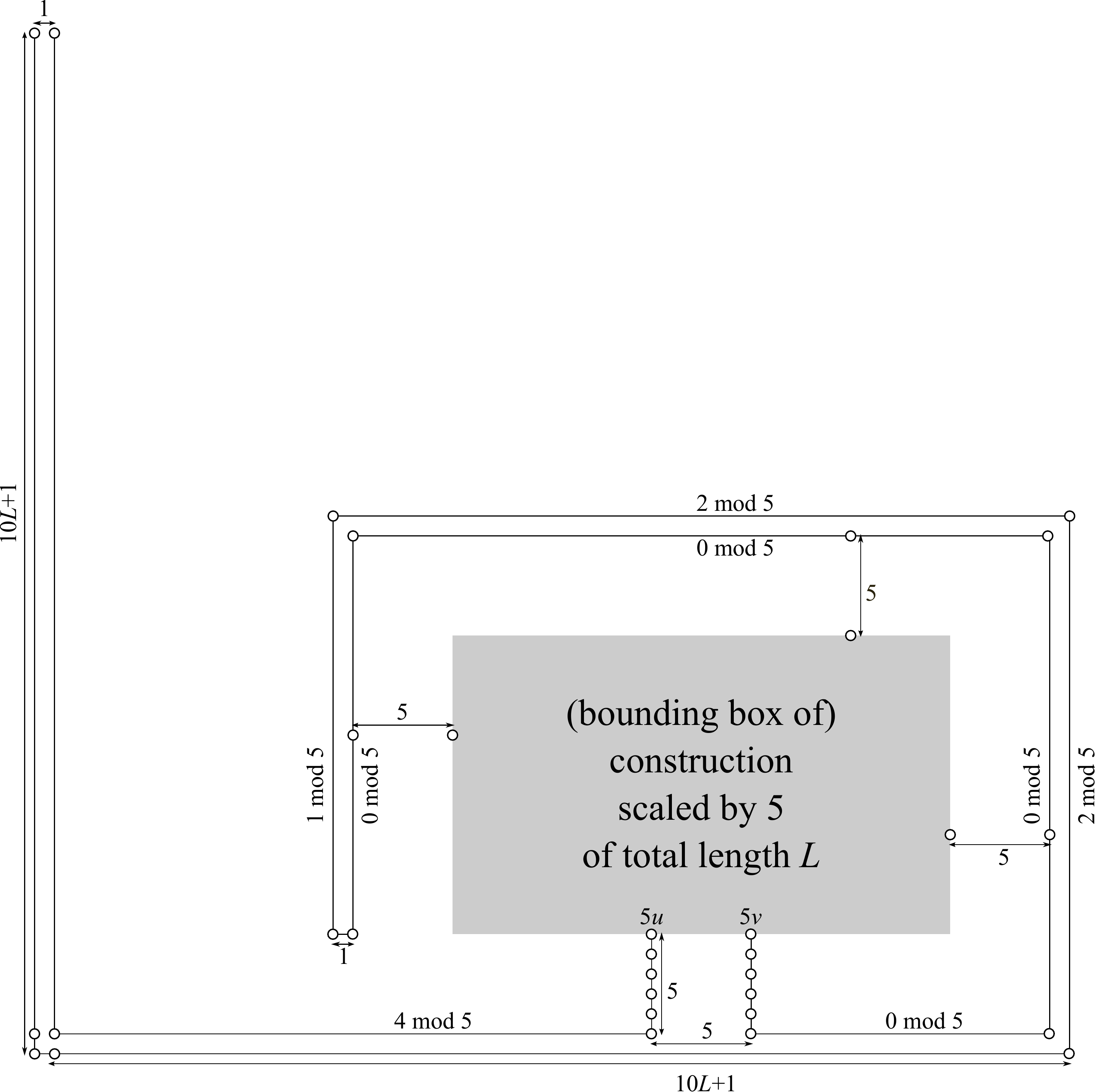}
    \caption{A frame gadget for an open chain which must fit in a $10L+1$ by $10L+1$ square.}
    \label{fig:frame_square}
\end{figure}

This frame gadget starts the chain with two straight segments of length
$s = 10 L + 1$.
Any embedding into the $s \times s$ square must place these segments along
two boundary edges of the square, say left and bottom as in the figure.
The next two segments on the outside of the frame gadget
must turn left to remain within the square.
At the other end of the chain, we have a vertical (by parity) segment
of length $s-1$, which forces its endpoints to be at the very top and
one position up from the bottom (to avoid overlapping the second segment).

Now we make a parity argument modulo~$5$.
Because the third segment goes up by $2$ modulo~$5$,
and the only other vertical travel modulo $5$ is by the fifth segment
which goes up or down by $\pm 1$,
the fifth segment must in fact go down by $1$ modulo~$5$
to reach the position one up from the bottom.

Because the next-to-last segment has length $>9L$,
its right end must be inside the frame.
Thus the sixth segment of length $1$ must go right to stay inside the frame;
otherwise, it could never connect to the right end of the next-to-last segment.
The rest of the segments are then forced to turn as in the figure
in order to avoid collisions.

The chain has length at most $48L$ (from the given chain of length $L$, the smaller frame, and the three long bars). Thus the length is $O(s)$.
\end{proof}

It remains open whether the problem of \defn{densely} packing
a fixed-angle orthogonal equilateral open chain of length $s^2$
into an $s \times s$ square is NP-complete, meaning that the chain
covers \emph{all} grid points in the square.
The analogous problem in 3D is strongly NP-complete \cite{AbelDemaineDemaineEisenstatLynchSchardl2013}.

\section*{Acknowledgments}
This work was initiated at the 3rd Virtual Workshop on Computational Geometry held in March 2022.
We thank the other participants of that workshop --- in particular Martin Demaine, David Eppstein, Timothy Gomez, and Aaron Williams ---
for helpful discussions and for providing a fruitful collaborative environment.

\bibliography{main}

\end{document}